\documentclass[letterpaper,12pt]{amsart} %%LaTeX 2e
\usepackage{amsmath,amsthm,amssymb,amsfonts}
\usepackage{graphicx,float,mathcomp,epsfig}
\usepackage{subfigure}
\input xy
\xyoption{all}

\renewcommand{\i}{\mathrm{i}} % imaginary

\renewcommand{\vec}[1]{\boldsymbol{#1}}
\renewcommand{\eqref}[1]{Eq.(\ref{#1})}
\newcommand{\commentout}[1]{}

\usepackage{amsmath}
\usepackage{amssymb}
\usepackage{mathrsfs}
\usepackage{graphicx}
\usepackage{graphicx}% Include figure files
\usepackage{dcolumn}% Align table columns on decimal point
\usepackage{bm}% bold math
\usepackage{amsfonts}
\usepackage{latexsym}
\usepackage{pdfsync}
\usepackage{colordvi}
\usepackage{color}

\begin{document}

\newcommand{\nwc}{\newcommand}
\nwc{\red}{\color{red}}
\newcommand{\bz}{{\mathbf z}}
\newcommand{\sqk}{\sqrt{\ks}}
\newcommand{\sqkone}{\sqrt{|\ks_1|}}
\newcommand{\sqktwo}{\sqrt{|\ks_2|}}
\newcommand{\invsqkone}{|\ks_1|^{-1/2}}
\newcommand{\invsqktwo}{|\ks_2|^{-1/2}}
\newcommand{\partz}{\frac{\partial}{\partial z}}
\newcommand{\grady}{\nabla_{\by}}
\newcommand{\gradp}{\nabla_{\bp}}
\newcommand{\gradx}{\nabla_{\bx}}
\newcommand{\invf}{\cF^{-1}_2}
\newcommand{\myphi}{\Phi_{(\eta,\rho)}}
\newcommand{\minrg}{|\min{(\rho,\gamma^{-1})}|}
\newcommand{\al}{\alpha}
\newcommand{\xvec}{\vec{\mathbf x}}
\newcommand{\kvec}{{\vec{\mathbf k}}}
\newcommand{\lt}{\left}
\newcommand{\ksq}{\sqrt{\ks}}
\newcommand{\rt}{\right}
\nwc{\bG}{{\bf G}}
\newcommand{\ga}{\gamma}
\newcommand{\vas}{\varepsilon}
\newcommand{\lan}{\left\langle}
\newcommand{\ran}{\right\rangle}
\newcommand{\tvas}{{W_z^\vas}}
\newcommand{\psiep}{{W_z^\vas}}
\newcommand{\wep}{{W^\vas}}
\newcommand{\weptil}{{\tilde{W}^\vas}}
\newcommand{\wepz}{{W_z^\vas}}
\newcommand{\weps}{{W_s^\ep}}
\newcommand{\wepsp}{{W_s^{\ep'}}}
\newcommand{\wepzp}{{W_z^{\vas'}}}
\newcommand{\wepztil}{{\tilde{W}_z^\vas}}
\newcommand{\vvas}{{\tilde{\ml L}_z^\vas}}
\newcommand{\veptil}{{\tilde{\ml L}_z^\vas}}
\newcommand{\vep}{{{ V}_z^\vas}}
\newcommand{\cvc}{{{\ml L}^{\ep*}_z}}
\newcommand{\cvcp}{{{\ml L}^{\ep*'}_z}}
\newcommand{\cvp}{{{\ml L}^{\ep*'}_z}}
\newcommand{\cvtil}{{\tilde{\ml L}^{\ep*}_z}}
\newcommand{\cvtilp}{{\tilde{\ml L}^{\ep*'}_z}}
\newcommand{\vtil}{{\tilde{V}^\ep_z}}
\newcommand{\ktil}{\tilde{K}}
\newcommand{\n}{\nabla}
\newcommand{\tkappa}{\tilde\kappa}
\newcommand{\ks}{{\omega}}
\newcommand{\bx}{\mb x}
\newcommand{\br}{\mb r}
\nwc{\bH}{{\mb H}}
\newcommand{\bu}{\mathbf u}
\nwc{\bxp}{{{\mathbf x}}}
\nwc{\byp}{{{\mathbf y}}}
\newcommand{\bD}{\mathbf D}
\nwc{\bh}{\mathbf h}
\nwc{\bn}{\mathbf n}
\newcommand{\bB}{\mathbf B}
\newcommand{\bC}{\mathbf C}
\nwc{\cO}{\mathcal  O}
\newcommand{\bp}{\mathbf p}
\newcommand{\bq}{\mathbf q}
\newcommand{\by}{\mathbf y}
\nwc{\bP}{\mathbf P}
\nwc{\bs}{\mathbf s}
\nwc{\bX}{\mathbf X}
\newcommand{\pdg}{\bp\cdot\nabla}
\newcommand{\pdgx}{\bp\cdot\nabla_\bx}
\newcommand{\one}{1\hspace{-4.4pt}1}
\newcommand{\corr}{r_{\eta,\rho}}
\newcommand{\rinf}{r_{\eta,\infty}}
\newcommand{\rzero}{r_{0,\rho}}
\newcommand{\rzeroinf}{r_{0,\infty}}
\nwc{\om}{\omega}
% theorem-like enviroments:
\nwc{\Gp}{{G_{\rm par}}}
\nwc{\nwt}{\newtheorem}
\nwc{\xp}{{x^{\perp}}}
\nwc{\yp}{{y^{\perp}}}
\nwt{remark}{Remark}
\nwt{definition}{Definition} %def is already defined
\nwc{\bd}{{\mb d}}
\nwc{\ba}{{\mb a}}
\nwc{\bal}{\begin{align}}
\nwc{\be}{\begin{equation}}
\nwc{\ben}{\begin{equation*}}
\nwc{\bea}{\begin{eqnarray}}
\nwc{\beq}{\begin{eqnarray}}
\nwc{\bean}{\begin{eqnarray*}}
\nwc{\beqn}{\begin{eqnarray*}}
\nwc{\beqast}{\begin{eqnarray*}}

%\nwc{\ea}{\end{array}}
\nwc{\eal}{\end{align}}
\nwc{\ee}{\end{equation}}
\nwc{\een}{\end{equation*}}
\nwc{\eea}{\end{eqnarray}}
\nwc{\eeq}{\end{eqnarray}}
\nwc{\eean}{\end{eqnarray*}}
\nwc{\eeqn}{\end{eqnarray*}}
\nwc{\eeqast}{\end{eqnarray*}}

\nwc{\ep}{\varepsilon}
\nwc{\eps}{\varepsilon}
\nwc{\ept}{\epsilon}
\nwc{\vrho}{\varrho}
\nwc{\orho}{\bar\varrho}
\nwc{\ou}{\bar u}
\nwc{\vpsi}{\varpsi}
\nwc{\lamb}{\lambda}
\nwc{\Var}{{\rm Var}}

\nwt{cor}{Corollary}
\nwt{proposition}{Proposition}
\nwt{corollary}{Corollary}
\nwt{theorem}{Theorem}
\nwt{summary}{Summary}
\nwt{lemma}{Lemma}
%\nwt{definition}{Definition}
\nwc{\nn}{\nonumber}
%\nwc{\bm}{\boldmath}
\nwc{\mf}{\mathbf}
\nwc{\mb}{\mathbf}
\nwc{\ml}{\mathcal}
\nwc{\bj}{{\mb j}}
\nwc{\bA}{{\mb \Phi}}
\nwc{\IA}{\mathbb{A}} %algebraic
\nwc{\bi}{\mathbf i}
\nwc{\bo}{\mathbf o}
\nwc{\IS}{\mathbb{S}}
\nwc{\IC}{\mathbb{C}} %complex
\nwc{\ID}{\mathbb{D}} %Dedekind
\nwc{\IM}{\mathbb{M}} %Dedekind
\nwc{\IP}{\mathbb{P}} %Dedekind
\nwc{\bI}{\mathbf{I}} %Dedekind
\nwc{\IE}{\mathbb{E}} %Euklides
\nwc{\IF}{\mathbb{F}} %finite field
\nwc{\IG}{\mathbb{G}} %Gauss
\nwc{\IN}{\mathbb{N}} %natural
\nwc{\IQ}{\mathbb{Q}} %rational
\nwc{\IR}{\mathbb{R}} %real
\nwc{\IT}{\mathbb{T}} %torus
\nwc{\IZ}{\mathbb{Z}} %integers
\nwc{\IV}{\mathbb{V}}
\nwc{\IX}{\mathbb{X}}
\nwc{\IY}{\mathbb{Y}}

\nwc{\cE}{{\ml E}}
\nwc{\cP}{{\ml P}}
\nwc{\cQ}{{\ml Q}}
\nwc{\cL}{{\ml L}}
\nwc{\cX}{{\ml X}}
\nwc{\cW}{{\ml W}}
\nwc{\cZ}{{\ml Z}}
\nwc{\cR}{{\ml R}}
\nwc{\cV}{{\ml V}}
\nwc{\cT}{{\ml T}}
\nwc{\crV}{{\ml L}_{(\delta,\rho)}}
\nwc{\cC}{{\ml C}}
\nwc{\cA}{{\ml A}}
\nwc{\cK}{{\ml K}}
\nwc{\cB}{{\ml B}}
\nwc{\cD}{{\ml D}}
\nwc{\cF}{{\ml F}}
\nwc{\cS}{{\ml S}}
\nwc{\cM}{{\ml M}}
\nwc{\cG}{{\ml G}}
\nwc{\cH}{{\ml H}}
\nwc{\bk}{{\mb k}}
\nwc{\bT}{{\mb T}}
\nwc{\bM}{{\mb M}}
\nwc{\cbz}{\overline{\cB}_z}
\nwc{\supp}{{\hbox{\rm supp}}}
\nwc{\fR}{\mathfrak{R}}
\nwc{\bY}{\mathbf Y}
\newcommand{\mbr}{\mb r}
\nwc{\pft}{\cF^{-1}_2}
\nwc{\bU}{{\mb U}}
\nwc{\bPhi}{{\mb \Phi}}
\nwc{\bPsi}{{\mb \Psi}}
\nwc{\im}{{\rm i}}
\nwc{\bN}{{\mathbf N}}
\nwc{\bw}{{\mathbf w}}
\nwc{\mbm}{{\mathbf m}}
\nwc{\lbr}{\textlbrackdbl}
\nwc{\rbr}{\textrbrackdbl}
\nwc{\vzero}{{\mathbf 0}}
\nwc{\cN}{{\mathcal N}}
\nwc{\rbra}{\textrbrackdbl}
\nwc{\lbra}{\textlbrackdbl}
\nwc{\conv}{\hbox{conv}}
\nwc{\rank}{\hbox{rank}}
\title{Absolute Uniqueness of  Phase Retrieval
with Random Illumination}

\author{Albert  Fannjiang}
\address{Department of Mathematics, UC Davis, CA 95616-8633.
fannjiang@math.ucdavis.edu}
  \thanks{ The research is partially supported by
the NSF grant DMS - 0908535.}
\maketitle
\begin{abstract}
Random  illumination is proposed to  enforce  absolute uniqueness and resolve all
types of ambiguity, trivial or nontrivial, in  phase retrieval. 
Almost sure irreducibility is proved
for {\em any} complex-valued object whose support set 
has rank $\geq 2$.   While  the new  irreducibility result
can be viewed as a probabilistic version  of the classical
result by Bruck, Sodin and Hayes, it provides a novel perspective and an effective method for phase retrieval. 
% {\em any} object of sufficiently high sparsity.  
%not just objects outside of a measure-zero set as in the classical result.
 In particular, almost sure uniqueness, up to
a global phase,   is  proved for complex-valued objects under  general two-point conditions.
Under a tight sector constraint  absolute uniqueness is proved to hold with probability exponentially close to unity as  the object sparsity increases.  
Under a magnitude constraint with random amplitude illumination, uniqueness modulo global phase is proved to hold with probability exponentially close to unity as object sparsity increases. For general complex-valued objects without any constraint, almost sure uniqueness up to global phase is
established  with two sets of  Fourier magnitude data 
under  two  independent illuminations. Numerical experiments
suggest  that  random illumination essentially alleviates most, if not all,  numerical problems commonly associated with the standard phasing algorithms.

\end{abstract}

\maketitle

\section{Introduction}
Phase retrieval is a fundamental problem in many areas of physical sciences
such as X-ray crystallography, astronomy, electron microscopy, coherent light microscopy, quantum state tomography and remote sensing. Because of  loss of
the phase information
a central question of phase retrieval is the uniqueness of solution
which is the focus of the present work. 

Researchers in phase retrieval, however, 
have long settled with 
%accepted  as inevitable 
the notion of {\em relative uniqueness}  (i.e. irreducibility) for generic (i.e. random) objects,  without
a practical means for deciding the reducibility of a given
(i.e. deterministic)  object,  and 
searched for  various ad hoc strategies to circumvent  
problems with stagnation and error in reconstruction. 
The common problem of stagnation may be  due to  the possibility of the iterative process to approach the object and  its twin or shifted image, the support not tight enough or the boundary not sharp enough  \cite{Fie2,FW,He}. 
Besides  the uniqueness issue, phase retrieval is also inherently  nonconvex and 
many researchers have believed   the lack of convexity
in the Fourier magnitude constraint to be  a main, if not the dominant,  source of numerical 
problems with the standard phasing algorithms \cite{BCL, Mar2, Sta}. While there have been dazzling  advances   in 
applications of phase retrieval in the past decades 
\cite{MIS}, we still do not know
just how much of the error and stagnation problems  is
attributable to 
to the lack of  uniqueness or convexity. 
 
We propose here   to refocus on the issue
of uniqueness as uniqueness is undoubtedly the first  foundational  issue 
of any inverse problem, including phase retrieval.
 Specifically we will first establish 
uniqueness in the absolute sense with {\em random}  illumination  under 
general,  physically reasonable object constraints (Figure \ref{fig0})  and secondly 
demonstrate that   random illumination practically  alleviates most numerical problems and 
drastically improves the quality of reconstruction. 

\begin{figure}[t]
\centering 
  \subfigure[]{
\includegraphics[width=4cm,height=4cm]{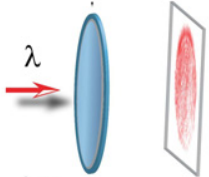}} \hspace{3cm}
  \subfigure[]{
\includegraphics[width=4.3cm,height=4.3cm]{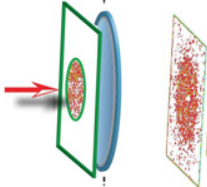}
}
\label{fig0}
\caption{Illumination of a partially transparent object (the blue oval) with a  deterministic (a) or random field $\lambda$  created by a diffuser  (b) followed by an intensity measurement of the diffraction pattern. In the case of wave front reconstruction,
the random modulator is placed at the exit pupil
instead of the entrance pupil as in (b). %(adapted from \cite{Alm}).
}
\end{figure}

To fix the idea, consider the discrete version of the  phase retrieval problem:
Let $\bn=(n_1,\cdots, n_d)\in \IZ^d$ and $\bz=(z_1,\cdots, z_d)\in \IC^d$. Define the multi-index notation $\bz^\bn=z_1^{n_1} z_2^{n_2} \cdots z_d^{n_d}$. Let $f(\bn)$ be a finite complex-valued function defined on $\IZ^d$  vanishing 
outside the finite lattice 
\[
\cN=\Big\{\vzero\leq\bn\leq \bN\Big\}
\]
for  $\bN=(N_1,\cdots, N_d)\in \IN^d$. We use the notation $\mbm\leq \bn$ for
$ m_j\leq n_j, \forall j$. The $z$-transform of a finite sequence
$f(\bn)$ is given by
\[
F(\bz)=\sum_\bn f(\bn) \bz^{-\bn}.
\]
The Fourier transform can be obtained from the $z$-transform as
\[
F(\bw)=F(e^{\im 2\pi w_1},\cdots,e^{\im 2\pi w_d})=\sum_{\bn} f(\bn) e^{-\im 2\pi \bn\cdot\bw},\quad \bw=(w_1,\cdots,w_d)\in [0,1]^d
\]
by some abuse of notation. 
The discrete phase retrieval problem is to determine 
$f(\bn)$ from the knowledge of the Fourier magnitude $|F(\bw)|, \forall \bw\in [0,1]^d$.

The question of uniqueness was partially answered  in \cite{Bat, BS, Hay, HM} which says that in dimension two or higher and with
the exception of  a measure zero
set of finite sequences   phase retrieval has a unique solution up to the equivalence class of ``trivial associates" (i.e. relative uniqueness).
These trivial, but omnipresent,  ambiguities include
 constant global phase,
 \[
 f(\cdot)\longrightarrow e^{\im\theta}f(\cdot),\quad \hbox{for some}\,\,\theta\in [0,2\pi],
 \]
 spatial shift
  \[
   f(\cdot)\longrightarrow f(\cdot +\mbm),\quad\hbox{for some}\,\,\mbm\in \IZ^d,
    \]
        %where $\bn + \mbm=\bn+\mbm (mod(N_1+1, \cdots, N_d+1))$, 
    and  conjugate inversion 
    \[
    f(\cdot)\longrightarrow f^*(\bN-\cdot).
    \]
 Conjugate inversion produces the so-called twin image. 

    This landmark uniqueness result, however,  does not address the following issues. First, a {\em given} object array,
  there is no way of deciding  {\em a priori} the irreducibility of
  the corresponding $z$-transform and the relative uniqueness of the phasing problem. 
%Indeed, arrays 
    %of hidden symmetries belong to
   % this unknown set of ambiguous sequences which challenges  the validity of the widely held assumption 
 %that    relative uniqueness holds true  in most of the practical problems.  
\commentout{
Also the analysis in \cite{BS, Hay, HM}  did not take into account  the practically relevant measure-zero set  of objects
 whose sparsity (the number of nonzero elements)  is strictly less than $|\cN|=\prod_{j=1}^d(1+N_j)$.  In other words, the uniqueness is guaranteed 
 only for objects whose strict  support is exactly $\cN$ and only
 up to the equivalence class of  global phase, spatial shift and
conjugate  inversion. } Secondly, although visually no different from 
the true image the trivial associates (particularly spatial shift and
conjugate inversion) nevertheless
``confuse" the standard numerical iterative processes
and cause serious  stagnation \cite{Fie2, FW,  MSC, Sta}. 
    %To remedy this difficulty, favorable initial inputs in addition to support and positivity constraints  are often needed \cite{Fie2, FW,  MSC}.
   % As pointed out in \cite{JJB},
    %the need for support constraint is one of the
   %most stringent limitations of the standard methods.

In this paper,  we study the notation of {\em absolute uniqueness}: if two finite objects $f$ and $g$  give rise
to the same Fourier magnitude data, then $f=g$ unequivocally.  
More importantly, 
we present the approach of random (phase or amplitude) illumination to the absolute uniqueness of phase retrieval. The idea of random illumination is related to coded-aperture imaging whose utility in other imaging contexts than phase retrieval  has  been established  experimentally \cite{AH,Alm,BWW,FAK,Rot,SMG} as well as  mathematically \cite{rand-illum, Rom}.  
%We show that oversampling \cite{Bat, MS, MSC}  of Fourier magnitudes with {\em single} random illumination removes, with probability one, 
%all ambiguities, including the trivial ones, for {\em every} complex-valued object satisfying certain  two-point properties  (Theorem \ref{cor1}), thus eliminating the need for support constraint. We also prove that the absolute uniqueness holds for objects with nonnegative real and imaginary parts with probability exponentially  close to
%unity as the object  sparsity increases. 

Our basic tool is  an improved  version   (Theorem \ref{thm:new}) 
of the irreducibility result of \cite{Hay, HM} with, however, a  completely different perspective and important  practical  implications. %The advantage of our probabilistic approach lies in that
%the measure is endowed in the ensemble of
% random illuminations, thus avoiding  
%the ambiguity  with the measure zero set of exceptional objects in the classical setting. 
The  main  difference is
that while the classical result \cite{Hay,HM} works with generic (thus random)
objects from a certain ensemble Theorem \ref{thm:new}
can deal with a {\em given, deterministic} object whose
support has rank $\geq 2$. 
  This improvement is achieved by
endowing  
the probability measure on the ensemble of illuminations, which we can manipulate,  instead of  the space of  objects, which we can not control, as in the classical setting.

On the basis of almost sure irreducibility, the mere assumption that the phases or magnitudes  of the object at two arbitrary points
lie in a countable set enforces 
uniqueness, up to a global phase,  in phase retrieval with a single random  illumination (Theorem \ref{cor1}). The absolute uniqueness can be enforced then
by imposing the positivity constraint (Corollary \ref{cor}).
For objects satisfying a tight sector condition, 
absolute uniqueness is valid with high probability depending
on the object sparsity for either phase or amplitude 
illumination (Theorem \ref{thm4}). For complex-valued objects under a magnitude constraint, uniqueness up to a global phase is valid with high probability (Theorem \ref{thm6}). For general complex-valued objects, almost sure uniqueness, up to global phase, is proved 
for phasing with  two independent illuminations (Theorem \ref{thm5}).

%For complex-valued objects, we show that oversampling of Fourier magnitudes with {\em two} random (phase and/or amplitude)
%illumination removes, with probability one, all ambiguities except
%for a global phase (Theorem \ref{thm3}).

The paper is organized as follows. In Section \ref{sec2} we discuss
various sources of ambiguity. In Section \ref{sec3} we prove
the almost sure irreducibility (Theorem \ref{thm:new} and Appendix). In Section \ref{sec4}
we  derive the uniqueness results (Theorem \ref{cor1},  \ref{thm4}, \ref{thm6},  \ref{thm5} and Corollary \ref{cor}). 
%In Section \ref{sec4} we derive the uniqueness result for complex-valued objects with two random illuminations. 
We demonstrate phasing with random illumination in Section \ref{sec:num}.  
We conclude
in Section \ref{sec5}. 
\section{Sources of ambiguity}\label{sec2}

As commented before the phase retrieval problem 
does not have a unique solution.  Nevertheless, the possible solutions are constrained as stated in the following theorem  \cite{Hay, PG}. 

\begin{theorem}

\label{thm1}
Let the $z$-transform $F(\bz)$ of a finite complex-valued  sequence $\{f(\bn)\} $ %vanishing outside the region ${\mathbf 0}\leq \bn \leq \bN$
 be
given by
\beq
F(\bz)=\alpha \bz^{-\mbm} \prod_{k=1}^p F_k(\bz),\quad  \mbm\in \IN^d, \alpha\in \IC\label{21}
\eeq
where $F_k, k=1,...,p$ are 
nontrivial irreducible polynomials. Let $G(\bz)$ be
the $\bz$-transform of another finite sequence $g(\bn)$. 
%vanishing outside ${\mathbf 0}\leq \bn\leq \bN$.
Suppose 
$|F(\bw)|=|G(\bw)|,\forall \bw\in [0,1]^d$. Then $G(\bz)$ must have
the form
\beqn
G(\bz)=|\alpha| e^{\im \theta} \bz^{-\bp}
\lt(\prod_{k\in I} F_k(\bz)\rt)
\lt(\prod_{k\in I^c} F_k^*(1/\bz^*)\rt),\quad\bp\in \IN^d,\theta\in \IR
\eeqn
 where $I$ is a subset of $\{1,2,...,p\}$. 
\end{theorem}
To start,  it is convenient to write
 \beq
  |F(\bw)|^2&=& \sum_{\bn =-\bN}^{\bN}\sum_{\mbm\in \cN} f(\mbm+\bn)f^*(\mbm)
   e^{-\im 2\pi \bn\cdot \bw}\nn\\
    &=& \sum_{\bn =-\bN}^{\bN}\cC_f(\bn)
     e^{-\im 2\pi \bn\cdot \bw}\label{Fou}
      \eeq
       where
        \beq
	 \label{aut}
	  \cC_f(\bn)=\sum_{\mbm\in \cN} f(\mbm+\bn)f^*(\mbm)
	  \eeq
	  is the autocorrelation function of $f$. Note the
	 symmetry $\cC^*_f(\bn)=\cC_f(-\bn)$. 
	  
The theorem then follows straightforwardly from
the equality between the autocorrelation functions of $f$ and $g$,  because  $F(\bw)F^*(\bw)=G(\bw)G^*(\bw)$,  
and the unique factorization of polynomials (see \cite{PG} for more details). 

\begin{remark}\label{rmk1}
If the finite array $f(\bn)$ is known {\em a priori} to vanish
outside the lattice $\cN$, % ${\mathbf 0}\leq \bn \leq\bN$, 
then
by Shannon's sampling theorem for band-limited functions
the sampling domain for $\bw$ can be limited to
the finite regular grid 
\beq
\label{sample}
\cM=\Big\{(k_1, \cdots, k_d): \forall j=1,\cdots, d\,\, \&\,\,  k_j=0, {1\over 2N_j+1}, {2\over 2N_j+1}, \cdots,  {2N_j\over 2N_j+1}. \Big\}
\eeq
since $|F(\bw)|^2$ is  band-limited to the set $-\bN\leq \bn \leq\bN$. %When the finite array $f$ has fewer nonzero elements
%than $\prod_{j} (N_j+1)$, 
\end{remark}

\commentout{
\begin{remark}

When the autocorrelation function $C_f$ is sparse and has, say,
$K$ non-zero elements, then, with a high probability,   the whole function $|F(\bw)|^2, \bw\in \cM$ can be recovered by compressed sensing techniques \cite{CRT1,  Rau} from 
a sampling subset consisting of $\cO(K\sum_{j=1}^d\log N_j)$ independent, uniformly distributed
points  in $\cM$ or $[0,1]^d$.  % (see Section \ref{sec:cs}). 
%\[
%\Big\{(n_1, \cdots, n_d): \forall j=1,\cdots,d, \,\, n_j=0,1,\cdots, N_j\Big\}
%\]

If the array $f$ has $S\leq |\cN|$ nonzero components (i.e. the sparsity equals $S$), then
it is easy to see that $C_f$ has at most $S(S-1)/2+1$ nonzero
components. Hence   the whole function $|F(\bw)|^2, \bw\in \cM$ 
can be recovered from $\cO(S^2\sum_{j=1}^d\log N_j)$ samples  with high probability. 
\end{remark}
}

There are three sources of ambiguity. First, 
the  linear phase term $\bz^{-\mbm}$ in (\ref{21})  remain undetermined  because
the autocorrelation operation destroys information about
spatial shift. The unspecified constant phase $\theta$ is another
source of ambiguity. 
%In this paper, a solution is referred to as unique

\commentout{
To understand the nonuniqueness of phase retrieval, let us first certain symmetries of the solutions. 

Following \cite{PG},  a finite array $f(\bn)$ is said to be {\em conjugate symmetric}  if
%, for some vector $\bk$ of positive integers,
\beq
F(\bz)=\pm \bz^{-\bN} F^*(1/\bz^*).
\eeq
%A conjugate symmetric array has a Fourier transform with linear phase factor.
}

To understand the physical meaning of  the operation 
\[
F(\bz)\longrightarrow \bz^{-\bN} F^*(1/\bz^*)
\]
consider the case $d=1$
\[
z^{-N} F^*(1/z^*) =f^*(0) z^{-N} +f^*(1) z^{1-N}+\cdots + f^*(N)
\]
which is the $z$-transform of  the conjugate space-inversed array $\{f^*(N),f^*(N-1),\cdots, f^*(0)\}$. The same is true
in multi-dimensions. 

The subtlest  form of ambiguity is caused by {\em partial}  conjugate inversion on some, but not all, factors of a  factorable  object,  with a reducible $z$-transform,    without which the conjugate inversion, like  spatial shift and global  phase,  is 
global in nature and  considered  ``trivial" in the literature (even though the twin image may have an opposite orientation).
%In numerical reconstruction, the trivial ambiguities have to be eliminated by
%assuming favorable {\em a priori} knowledge such as support constraints and positivity.  

In this paper, we consider both types, trivial and nontrivial,
of ambiguity, as they both can degrade the performance of 
phasing   schemes. 
Our  main  purpose  is to show  by rigorous analysis that with  random
illumination it is possible to eliminate all 
ambiguities at once. 
%with the exception of global phase factor,   in phase retrieval. 
%To our knowledge all previous
%results, with the exception of \cite{CES},  on phase retrieval 
%concern only the reducibility/irreducibility issue. We shall refer to
%this as uniqueness up to equivalence class. In the present
%work we prove the absolute sense of uniqueness 
%by using two independent random illuminations to eliminate
% all three ambiguities. 

\section{Irreducibility}\label{sec3}  
Random illumination amounts to replacing the
original object $f(\bn)$ by 
\beq
\label{1}
\tilde f(\bn) =f(\bn) \lamb(\bn)
\eeq
where $\lamb(\bn)$, representing the incident field, is a {\em known} array of samples of random variables (r.v.s).
%typically assumed to be independent and identically distributed (i.i.d.).
The idea is to first modify the object by the encoding array $\lamb(\bn)$ so that phase retrieval has unique solution and then
use the prior knowledge of $\lamb$ to recover $f$. 

Nearly independent random illumination can be produced
by a diffuser placed near the  object, cf. Figure \ref{fig0}.
 The illumination field can 
be randomly modulated in {\em phase} only  with the use  computer generated holograms \cite{BWW}, random phase plates \cite{AH,Rot} and
liquid crystal phase-only panels \cite{FAK}. 
One of the best known amplitude masks is
uniformly redundant array \cite{URA} and its variants \cite{MURA}.
The advantage of phase mask, compared to amplitude mask, is the lossless energy transmission of an incident wavefront through the mask.  By placing
either phase or amplitude mask at a distance
from the object, one can create an illumination field
modulated in both amplitude and phase in a way dependent on the distance \cite{SMG}.

Let $\lamb(\bn)$ be  {\em continuous}  r.v.s  with respect to the Lebesgue measure on $\IS^1$ (the unit circle), $\IR$ or $\IC$.  The  case of 
$\IS^1$ can be facilitated by a random phase modulator with
\beq
\label{rph}
\lambda(\bn)=e^{\im \phi(\bn)}
\eeq
where $\phi(\bn)$ are  continuous r.v.s on $[0,2\pi]$ while the case of $\IR$ can be facilitated by a random amplitude
modulator. The case of $\IC$ involves simultaneously 
both phase and amplitude modulations.  More generally,
$\lambda(\bn)$ can be any continuous r.v.  on a real algebraic  variety $\cV(\bn)\subset \IC \simeq\IR^2$. For example  $\IR$ and $\IS^1$ can be viewed as  real projective  
varieties defined by the polynomial equations $y=0$
and $x^2+y^2-1=0$, respectively, on the complex plane identified as $\IR^2$.

The support  $\Sigma$ of a polynomial
$F(\bz)$  is the set of exponent vectors in $\IN^d$ with
nonzero coefficients. The rank of the support set
 is the dimension of its convex hull.

%For simplicity of
%notation, we shall consider the case of $d=2$ for the following result.
\begin{theorem}
Let $ \{f(\bn)\}$ be a finite complex-valued array whose support has rank $\geq 2$ and touches 
all the coordinate hyperplanes $\{n_j=0:j=1,\cdots,d\}$.
%both horizontal and vertical axes. 
Let $\{\lambda(\bn)\}$ be continuous r.v.s on  nonzero real algebraic  varieties  $\{\cV(\bn)\}$ in $\IC (\simeq \IR^2)$ 
with an {\em absolutely continuous}  joint distribution with respect to the standard product measure on $\prod_{\bn\in \Sigma}\cV(\bn)$ where
$\Sigma\subset \IN^d$ is the support set of $ \{f(\bn)\}$.  
%Let $\lamb(\bn)$ be i.i.d. {\em (absolutely) continuous}  random variables with respect to the Lebesgue measure on $\IS^1$, $\IR$ or $\IC$. 
Then 
the $z$-transform of $\tilde f(\bn) =f(\bn) \lamb(\bn)$ is irreducible with probability one. 
\label{thm:new}  
\end{theorem}
\begin{remark}
If the object support does not touch all the coordinate hyperplanes, then 
the the irreducibility holds true, {\em up to 
some monomial  of $\bz$.} In view of Theorem \ref{thm1}
this is sufficient for our purpose. 
\end{remark}
\begin{remark}

The theorem does not hold if the  rank-2 condition fails.
For example, let $p(\bz)$ be any monomial and consider 
\beq
\label{exp1}
F(\bz)=\sum_{j} c_j p^j(\bz)
%c_{0}+c_{1} \bz^{(n,nm)},\quad n=2,3\cdots,\quad m=1,2,3,\cdots
\eeq
which is reducible for any $c_j \in \IC$, except when $F$ is a monomial, by the fundamental theorem of algebra (of one variable).   Another example is the {\em homogeneous}
polynomials of a  sum degree $N$
\beq
\label{exp2}
F(\bz)=\sum_{i+j=N} c_{ij} z_1^iz_2^j
\eeq
which is factorable by, again,  the fundamental theorem of algebra.
\end{remark}

The proof of Theorem \ref{thm:new} is given in the Appendix.

Theorem \ref{thm:new} improves in several aspects on  the classical result   that
the set of the reducible polynomials has zero measure in the space of multivariate polynomials with real-valued
coefficients \cite{Hay, HM}. The  main  improvement is
that while the classical result works with generic (thus random)
objects  Theorem \ref{thm:new}
 deals  with {\em any deterministic} object with minimum (and necessary) conditions on its support set.
By definition,  deterministic objects belong to the
measure zero set excluded  in the classical setting of \cite{Hay,HM}.
 It is both theoretically and practically important that  Theorem \ref{thm:new} places 
the probability measure on the ensemble of illuminations, which we can manipulate,  instead of  the space of  objects, which we can not control.

\commentout{Moreover, this shift
of measurable space prevents the practically important 
set of sparse  objects from falling
through the ``cracks" of measure-zero sets in the formulation of \cite{Hay,HM}. Indeed, 
the sparsity constraint in the object domain can be an effective tool for improving 
the performance of reconstruction \cite{He}.

Even for
objects whose support is exactly the finite grid in the object domain,  it is impossible to know {\em a priori} if a given 
object lies in the exceptional  set of ambiguous objects.
}

\commentout{
\begin{remark}
\label{rmk2}
The assumption of independency of $\lamb(\bn)$ for different $\bn$ is used only tangentially in the above argument. What is
needed for the proof is the mere fact that the induced probability measure on $\IM_{\widetilde F}$ is absolutely continuous with 
respect to the Lebesgue measure. This  allows for the possibility of correlated and differently distributed illumination $\lamb(\bn)$.

\end{remark}
}

In the next section, we go further to show
that with additional, but for all practical purposes sufficiently general, constraints on the values of the object, we
can essentially remove all ambiguities with the only possible
exception of global phase factor. This decisive step distinguishes our method
from the standard approach. 

\section{Uniqueness}
\label{sec4}
Without additional {\em a priori} knowledge on the object  Theorem \ref{thm:new}, however,  does not preclude the trivial ambiguities  such
as global phase,  spatial shift and conjugate inversion. For example, we can produce
another finite array $\{g(\bn)\}$ 
%vanishing outside ${\mathbf 0}\leq\bn\leq\bN$ 
that yields the same measurement data by setting
\beq
\label{10}
g(\bn)&=& e^{\im \theta} f(\bn + \mbm )\lamb(\bn + \mbm)/\lamb(\bn)
\eeq
or
\beq
\label{11}
g(\bn)&=& e^{\im \theta}  f^*(\bN-\bn + \mbm)\lambda^* (\bN-\bn + \mbm)/\lamb(\bn)
\eeq
for $\theta\in [0,2\pi]$ and  $\mbm\in \IZ^d$.  
%where $\bn + \mbm=\bn+\mbm \,\,(\hbox{mod} (N_1+1,N_2+1))$. 
Expression (\ref{10}) and (\ref{11}) are  the remaining 
ambiguities to be addressed. 

\subsection{Two-point constraint}
One important  exception  is the case of  {\em real-valued} objects
when the illumination is  complex-valued (the case of $\IS^1$ or $\IC$).
In this case,  on the one hand (\ref{10}) produces a 
complex-valued array with probability one unless $\mbm=0,\theta=0, \pi$
and, on the other hand, (\ref{11}) is  complex-valued with probability one
regardless of $\mbm$. In this case, none of the trivial ambiguities can arise. Indeed,
a stronger result is true depending on the nature of random illumination.  

 \begin{theorem}\label{cor1} Suppose the
 object support has rank $\geq 2$.
Suppose either 
 of the following cases holds:\\

 (i) The phases of the object $\{f(\bn)\}$ at  two points, where $f$ does not vanish,  belong to
 a known countable subset of $[0,2\pi]$ and $\{\lamb(\bn)\}$ are independent continuous r.v.s on
 real algebraic varieties  in $\IC$ such that their angles are
continuously distributed  on $[0,2\pi]$ (e.g. $\IS^1$ or $\IC$) \\

 (ii) The amplitudes of the object $\{f(\bn)\}$ at two points, where $f$ does not vanish, 
 belong to a known measure zero  subset of $\IR$  and $\{\lamb(\bn)\}$ are independent continuous r.v.s on real algebraic varieties in $\IC$ such that their magnitudes
 are continuously distributed  on $(0,\infty)$ (e.g. 
 $\IR$ or $\IC$). \\
 
%  (iii) The amplitudes of the object $\{f(\bn)\}$ at  two points
% belong to a known measure zero set and that $\{\lamb(\bn)\}$ are independent continuous random variables on
% $\IR$.\\
 
Then $f$ is determined uniquely,   up to a global phase,   by the
  Fourier magnitude measurement  on the lattice $\cM$
     with probability one.
 \end{theorem}
 
 \begin{remark}
For  the two-point constraint in case (i)  to  be convex, it is necessary
for the constraint set to be a singleton, namely 
the phases of the object at two nonzero points must take
on a single known value. On the
other hand, the amplitude constraint in case (ii) can never
be convex unless the set is a singleton and the object phases are the same at the two points. 
\end{remark}

 \begin{proof} By Theorem \ref{thm:new} the $z$-transform of $\{\lambda(\bn)f(\bn)\}$ is irreducible   with probability one. We prove the theorem case by case.\\
 
\noindent { Case (i)}:  Suppose the phases of $f(\bn_1)$ and $f(\bn_2)$ belong to the coutable set $\Theta\subset [0,2\pi]$. Let us show the probability  that the phase of $g(\bn)$ as given by (\ref{10}) with $\mbm\neq 0$ takes on a value in $\Theta$ at two distinct points is zero.
 
 Since $\lamb(\bn+\mbm),\mbm\neq 0, $ and the phases of 
 $\lamb(\bn)$ are independent, continuous r.v.s  on $[0,2\pi]$, the phase of $g(\bn), \forall\bn,$ is continuously distributed on $[0,2\pi]$ for
 all $\theta$.
 
 Now suppose the phase of $ g(\bn_0)$ for some $\bn_0$ lies
 in the set $\Theta$. This implies that $\theta$ must belong
 to the countable set $\Theta'$ which is  $\Theta$ shifted by
 the negative phase of $f(\bn_0+\mbm)\lamb(\bn_0 + \mbm)/\lamb(\bn_0)$.
The phase of $g(\bn)$ at a different location $ \bn\neq\bn_0 $,
however,  
 almost surely does not take on any value in the set $\Theta$ for any fixed $\theta\in \Theta'$ unless $\mbm=0$. 
 Since a  countable union of measure-zero sets has zero measure, the probability
 that the phases of $g$ at two points lie in 
$\Theta$ is zero if $\mbm\neq 0$. 

Likewise, $\lambda^* (\bN-\bn + \mbm)/\lamb(\bn), \forall\mbm, $ has a random phase that is continuously distributed on $[0,2\pi]$ and by the same argument the probability that the phases of $g$ as given by (\ref{11}) at two points lie in 
$\Theta$ is zero. \\

\noindent {Case (ii):}  Suppose the amplitudes of $f(\bn_1)$ and $f(\bn_2)$ belong to the measure zero set $\cA$. 
Since $\lamb(\bn+\mbm), \mbm\neq 0, $ and
 $\lamb(\bn)$ are independent and continuously distributed on $\IR$ or $\IC$, the amplitude  of $g(\bn)$ as given by (\ref{10}) 
  is  continuously distributed  on $\IR$ and hence the probability 
  that 
 the amplitude of $g(\bn)$ as given by (\ref{10})  belongs to $\cA$
at any $\bn$ is zero. 

Now consider $g(\bn)$ given by (\ref{11}). Suppose  that  the amplitude of $g(\bn_0)$ belongs to $\cA$ at some $\bn_0$.
This is possible only for $\bn_0={(\bN+\mbm)/2}$
in which case $g(\bn_0)=e^{\im \theta} f^*(\bn_0)$. 
 The amplitude of $g(\bn),\bn\neq \bn_0, $ has a continuous distribution on $\IR$ and zero probability to  lie in
$\cA$. 

The global phase $\theta$, however, can not
be determined uniquely in either case. 

 \commentout{
Without loss of generality, we may assume that $f(\bn_1), f(\bn_2) \in \IR$ for some $\bn_1\neq \bn_2$. The general case of
same or opposite phase can be reduced to the real-valued case
by a global phase factor. 
 
 We first prove that in the case of (\ref{10}) $\mbm$ must be ${\mathbf 0}$ almost surely. Otherwise, $g(\bn)$ almost surely has nonzero imaginary part for $\bn\neq \bn_1, \bn_2$ whenever $g(\bn)\neq 0$ except for at most one point $\bn_3\neq\bn_1, \bn_2$ in 
 which case
 $\theta$ is the negative phase of $f(\bn_3 +  \mbm)\lamb(\bn_3 + \mbm)/\lamb(\bn_3)$. 
 
 Now it is clear that with probability one at most one of  the three quantities $g(\bn_1), g(\bn_2)$ and $g(\bn_3)$ can be real-valued 
 which violates the assumption of real-valuedness at at least
 two distinct points.  
 
 The proof for the case of (\ref{11}) is similar and omitted here. 
}
\end{proof}

The global phase factor can be determined uniquely by  additional
constraint on the values of the object. For example,
the following result follows immediately from Theorem \ref{cor1} (i). 
 
\begin{corollary}\label{cor}
Suppose that  $\{f(\bn)\}$ is real and nonnegative and its
support has rank $\geq 2$. Suppose that 
$\{\lamb(\bn)\}$ are independent continuous r.v.s on
real algebraic varieties in $\IC$ such that their phases
are continuously distributed  on $[0,2\pi]$ (e.g.
$\IS^1$ or $\IC$).  Then 
$\{f(\bn)\}$ can be determined absolutely uniquely with probability one.
\end{corollary}

\begin{proof}
With a real, positive  object, the countable set for phase 
is the singleton $\{0\}$ and the global phase is uniquely
fixed. 
\end{proof}
 
 %\begin{remark}\label{rmk4}
 \commentout{
 The simplest application of Theorem \ref{cor1}  is for
 recovering 
 real-valued objects whose phases are limited to the countable 
 set $\{0, \pi\}$ and for which a random phase illumination (case (i)) forces  absolute
 uniqueness. On the other hand, random amplitude illumination
 does not help in the case of  real-valuedness contraint. 
 }

\subsection{Sector constraint}
More generally, we consider the {\em sector} constraint that
 the phases of $\{f(\bn)\}$ belong to $[a,b]\subset [0,2\pi]$.
 For example,
 the  class of complex-valued objects relevant to
 $X$-ray diffraction typically have  nonnegative real and imaginary parts where the real part 	is the effective	number of electrons coherently diffracting photons, and the imaginary part  represents the attenuation  \cite{MSC}. For such objects,
 $[a,b]=[0,\pi/2]$. 
 
 Generalizing the argument for Theorem \ref{cor1} we can prove the following. 
 \begin{theorem}\label{thm4}
 Suppose the object support has rank $\geq 2$. 
Let   the finite object $\{f(\bn)\}$ satisfy  the {\em sector} constraint that
 the phases of $\{f(\bn)\}$ belong to $[a,b]\subset [0,2\pi]$.  Let $S$ be  the sparsity (the number of nonzero elements) of the object.   \\

 (i) Suppose $\{\lambda(\bn)\}$ are independent, identically distributed (i.i.d.) continuous  r.v.s on real algebraic varieties in $\IC$ such that 
their  phases $\{\phi(\bn)\}$ 
 are  uniformly distributed  on $[0,2\pi]$ (e.g. the random phase illumination (\ref{rph})).
Then with  probability at least $1-|\cN| |b-a|^{[S/2]}(2\pi)^{-[S/2]}$  
 the object $f$  is uniquely
 determined, up to a global phase,  by the Fourier magnitude measurement. Here
  $\lbra S/2\rbra$ is the greatest integer at most  $S/2$. \\
 
 (ii) Consider the random amplitude illumination with i.i.d. continuous r.v.s  $\{\lambda(\bn)\}\subset \IR$ that  are equally likely  negative or positive, i.e.  $\IP\{\lamb(\bn)>0\}=\IP\{\lamb(\bn)<0 \}=1/2, \forall \bn$. Suppose $|b-a|\leq\pi$. Then with  probability at least  
$
 1- |\cN| 2^{-[(S-1)/2]}% \lt(1-\max\Big\{0, {b-a\over \pi}-1\Big\}\rt)^{[(S-1)/2]}
$
  the object
 $f$ is uniquely
 determined, up to a global phase,  by the Fourier magnitude measurement.

 In both cases, the global phase is uniquely determined  if the sector $[a,b]$ is tight
 in the sense that no proper {\bf interval} of $[a,b]$ contains
 all the phases of the object.
 \end{theorem}
 \commentout{
 \begin{remark}
If the positivity constraint is imposed  on either real or imaginary  part (but not both),
then the same analysis as above yields  the  probability at least  $1- 2^{-[(S-1)/2]}|\cN|$ of absolute uniqueness in case (i) and (ii). 
\end{remark}
}

 \begin{proof} 
 \noindent Case (i):  
 Consider first the expression (\ref{10}) with any $\mbm\neq 0$
 and the  $\lbra S/2\rbra$  independently distributed r.v.s  of $g(\bn)$ corresponding to $[S/2]$ nonoverlapping pairs of points $\{\bn,\bn + \mbm\}$. The probability for every such the phase of $g(\bn)$ to lie in the sector $[a,b]$ 
 is $|b-a|/(2\pi)$ for any $\theta$ and hence the probability for 
 all $g(\bn)$ with $ \mbm\neq 0,\theta\neq 0, $ to lie in the sector  is at most  $|b-a|^{[S/2]}(2\pi)^{-[S/2]}$. The union over $\mbm\neq 0$  of these events  has probability
 at most $|\cN| |b-a|^{[S/2]}(2\pi)^{-[S/2]}$. 
 
 Likewise  the probability  for 
all $g(\bn)$ given by (\ref{11})  to  lie  in the first quadrant for any $\mbm$ is at most $|\cN| |b-a|^{[S/2]}(2\pi)^{-[S/2]}$.\\

\noindent Case (ii): For (\ref{10}) with any $\mbm\neq 0$  the  $[S/2]$ independently distributed random variables $g(\bn)$ corresponding to $[S/2]$ nonoverlapping pairs of points $\{\bn,\bn + \mbm\}$, 
satisfy the sector constraint  with probability at most  $2^{-[S/2]}$
if $|b-a|\leq \pi$.
% or $(1-(b-a)/(2\pi))^{[S/2]}$ if $|b-a|>\pi$  for any $\theta$. 
Hence the probability that all $g(\bn)$ with $\mbm\neq0$  satisfy the sector constraint  is at most  $
|\cN| 2^{-[S/2]}. 
%\lt(1-\max\Big\{0, {b-a\over \pi}-1\Big\}\rt)^{[S/2]}.
$

For (\ref{11}) with $\theta=0$ and  any $\mbm$, $g(\bn_0)=f(\bn_0)$ at $
\bn_0=(\bN+\mbm)/2$  and hence $g(\bn_0)$ lies in the first quadrant with probability one. For $\bn\neq \bn_0$,  $g(\bn)$  satisfies the sector constraint with probability $1/2$ if $|b-a|\leq \pi$. %or with probability  $1-(b-a)/(2\pi)$ if $|b-a|>\pi$.
 Now the $[(S-1)/2]$  independently distributed r.v.s $g(\bn) $ corresponding to nonoverlapping pairs of points $\{\bn,\bn + \mbm\},\bn\neq \bn_0,$  satisfy the sector constraint  with probability at most $2^{-[(S-1)/2]}$ if $|b-a|\leq \pi$. 
% or $(1-(b-a)/(2\pi))^{[(S-1)/2]}$ if $|b-a|>\pi$. 
Hence the probability that
all $g(\bn)$ given by (\ref{11}) with arbitrary $ \mbm$  satisfy
the sector constraint  is at most $
|\cN| 2^{-[(S-1)/2]} %\lt(1-\max\Big\{0, {b-a\over \pi}-1\Big\}\rt)^{[(S-1)/2]}.
$.
 \end{proof}
 
 \commentout{% bounded ranged objects
 Another commonly available prior knowledge is
 the upper and lower bounds of the object magnitudes. 
 \begin{theorem}
 Suppose that $ |f(\bn)|\in [a,b]\subset (0,\infty)$ for $\bn$ in
 the object support
  and that $\{\lamb(\bn)\}$ are independent continuous random variables on
 $\IR$ or $\IC$.
 \end{theorem}
 }

\subsection{Magnitude constraint}

Likewise if the object satisfies a magnitude constraint  then we can  use random {\em amplitude}  illumination to enforce uniqueness (up to a global phase). 

\begin{theorem} Suppose that the object support
has rank $\geq 2$. 
Suppose that $K$ pixels of the complex-valued object $f$ satisfy the magnitude constraint $0<a\leq |f(\bn)|\leq b$ and that $\{\lambda(\bn)\}$ are i.i.d. continuous r.v.s on real algebraic varieties in $\IC$ with $\IP\{|\lambda(\bn)/\lambda(\bn')|> b/a\,\,\hbox{or}\,\,  |\lambda(\bn)/\lambda(\bn')|<a/b\}=1-p>0$ for $\bn\neq\bn'$. Then the object $f$ is determined uniquely, up
to a global phase, by the Fourier magnitude data on $\cM$, with probability at least $1-|\cN| p^{[(K-1)/2]}$. 
\label{thm6}
\end{theorem}

\begin{proof} The proof is similar to that for Theorem \ref{thm4}(ii).

For (\ref{10}) with any $\mbm\neq 0$  the  $[K/2]$ independently distributed random variables $g(\bn)$ corresponding to $[K/2]$ nonoverlapping pairs of points $\{\bn,\bn + \mbm\}$ satisfy $0<a\leq |g(\bn)|\leq b$
with probability less than $p^{-[K/2]}$ for any $\theta$. Hence the probability that $g(\bn)$ with $\mbm\neq0$
satisfy the magnitude constraint at $K$ or more points  is at most  $|\cN| p^{-[K/2]}$. 

For (\ref{11}) with any $\mbm$, $|g(\bn_0)|=|f(\bn_0)|$ at $
\bn_0=(\bN+\mbm)/2$  and hence $g(\bn_0)$ satisfies the magnitude constraint with probability one. For $\bn\neq \bn_0$, there is at most   probability $p$  for $g(\bn)$  to
satisfy the magnitude constraint. By independence, the $[(K-1)/2]$  independently distributed r.v.s $g(\bn) $ corresponding to nonoverlapping pairs of points $\{\bn,\bn + \mbm\},\bn\neq \bn_0,$  satisfy the magnitude constraint  with probability at most $p^{[(K-1)/2]}$. Hence the probability that
$g(\bn)$ given by (\ref{11}) with arbitrary $ \mbm$ 
satisfy the magnitude constraint at $K$ or more points is at most $|\cN| p^{[(K-1)/2]}$.

The global phase factor is clearly undetermined. 
\end{proof}

As in Theorem \ref{cor1} case (ii) the magnitude constraint here, however, is not convex.

\subsection{Complex objects without constraint}
 For general complex-valued objects without any constraint, we consider two sets of Fourier magnitude data produced with two  independent random illuminations and obtain
almost sure uniqueness modulo global phase.   

\begin{theorem} 

\label{thm5}

Let $ \{f(\bn)\}$ be a finite complex-valued array
whose support has rank $\geq 2$.   
%vanishing outside $\cN$. 
Let $\{\lamb_1(\bn)\}$ and $\{\lamb_2(\bn)\}$ be two independent arrays  of r.v.s  satisfying
the assumptions in Theorem \ref{thm:new}.
% on $(\IS^1)^{2|\cN|}$, $\IR^{2|\cN|}$ or $\IC^{2|\cN|}$. 
%Let
%$\widetilde F_1(\bw)$ and $\widetilde F_2(\bw)$ be the Fourier transform of $\tilde f_1(\bn)=f(\bn)\lamb_1(\bn)$ and $\tilde f_2(\bn)=f(\bn)\lamb_2(\bn)$, respectively. 
%Suppose that $|\widetilde F_1(\bw)|$ and $|\widetilde F_2(\bw)|$ 
%are known for all $ \bw\in \cM$. 
 
Then with probability one $f(\bn)$ is uniquely determined,  up to a global phase, by the Fourier magnitude measurements on $\cM$  with two illuminations $\lamb_1$ and $\lamb_2$. 
%The only exception is when the object $f$ has sparsity 1 and the random illuminations are purely real-valued. 

If the second illumination $\lamb_2$ is deterministic and results in 
an irreducible  $z$-transform 
while $\lamb_1$ is random as above, then
the same conclusion holds.
\end{theorem}
\begin{proof}
Let  $g(\bn)$ be another array that vanishes  outside $\cN$  and produces the same data. By Theorem \ref{thm1},  \ref{thm:new} and Remark \ref{rmk1}
\beq
\label{12}
g(\bn)&=&\lt\{\begin{matrix} e^{\im \theta_i} f(\bn + \mbm_i)\lamb_i(\bn + \mbm_i)/\lamb_i(\bn) \\
e^{\im \theta_i} f^*(\bN-\bn + \mbm_i)\lamb_i^* (\bN-\bn + \mbm_i)/\lamb_i(\bn),
\end{matrix}\rt.
\eeq
for  some $\mbm_i\in \IZ^d, \theta_i\in \IR, i=1,2$. 

Four scenarios of ambiguity exist but  because of  the independence of
$\lamb_1(\bn), \lamb_2(\bn)$ none can arise.

First of all, if
\[
g(\bn)=e^{\im \theta_i} f(\bn + \mbm_i )\lamb_i(\bn + \mbm_i)/\lamb_i(\bn),\quad i=1,2
\]
then
\[
e^{\im\theta_1}f(\bn + \mbm_1 )\lamb_1(\bn + \mbm_1)/\lamb_1(\bn)
=e^{\im \theta_2}  f(\bn + \mbm_2 )\lamb_2(\bn + \mbm_2)/\lamb_2(\bn).
\]
This almost surely can not occur unless
$\mbm_1=\mbm_2={\mathbf 0}, \theta_1=\theta_2$ in which case
$g$ equals  $f$ up to a global phase factor. 

The other possibilities  can be similarly ruled out:
\beqn
g(\bn)&=&e^{\im \theta_1} f(\bn + \mbm_1 )\lamb_1(\bn + \mbm_1)/\lamb_1(\bn)\\
&=&e^{\im \theta_2} f^*(\bN-\bn + \mbm_2)\lamb_2^* (\bN-\bn + \mbm_2)/\lamb_2(\bn)
\eeqn
and
\beq
\label{exp}
g(\bn)=e^{\im \theta_i} f^*(\bN-\bn + \mbm_i)\lamb_i^* (\bN-\bn + \mbm_i)/\lamb_i(\bn),\quad i=1,2
\eeq
for any $\mbm_i,\theta_i, i=1,2$. 
%The only exception occurs in the scenario (\ref{exp}) when $f(\bn)=0,\forall \bn\neq \bn_0$ and the illuminations are real-valued. 

\commentout{
Now consider the case that $\{\lamb_2(\bn)\}$ is deterministic. Let $g$ be given as in (\ref{12}) with $i=1$.
Then the Fourier magnitude data for the second
illumination 
\beq
\label{2.12}
\widetilde F_2(\bw)=\sum_{\bn} g(\bn)\lambda_2(\bn) e^{-\im 2\pi \bn\cdot\bw}
%\bPhi \Lambda_2\lt\{\begin{matrix} e^{\im \theta_1} f(\bn + \mbm_1)\lamb_i(\bn + \mbm_1)/\lamb_1(\bn) \\
%e^{\im \theta_1} f^*(\bN-\bn + \mbm_1)\lamb_1^* (\bN-\bn + \mbm_1)/\lamb_1(\bn),
%\end{matrix}\rt.
\eeq
%where $\bPhi$ is the discrete Fourier transform
%and $\Lambda_2=\hbox{diag} (\lamb_2(\bn))$.  
are {\em continuous}  r.v.s  unless
$g(\bn)=e^{\i\theta_1}f(\bn)$. On the other
hand, since $\{\lamb_2(\bn)\}$ is deterministic, 
the Fourier magnitude $\widetilde F_2(\bw)$  must be deterministic also.
Thus, $g(\bn)=e^{\i\theta_1}f(\bn)$ for some constant $\theta_1$. }
The same argument above applies to the case of deterministic 
$\lambda_2$ if the resulting $z$-transform is irreducible. 

\end{proof}

\commentout{
\begin{remark}
\label{rmk3}
%Similar to  the relaxation of the independency assumption in Remark \ref{rmk2},  
The two illuminations $\lamb_1$ and $\lamb_2$ need not
be independent from each other. It suffices for
 $\lamb_1(\bn)$ and $\lamb_2(\mbm)$, for any $\bn\neq \mbm$,  to have a joint
 distribution that is absolutely continuous with respect to
 the product measure. 

\end{remark}
}

\commentout{
\section{Compressed sensing for sparse objects}
\label{sec:cs}

When the object $f(\bn)$ is sparse in the sense that
 its support set has fewer  elements than 
$(N_1+1)(N_2+1)$, then it is possible to  to recover
 the whole data set $|F(\bw)|^2, \bw\in \cM$ by measuring 
 fewer samples than $\cM$. 
 
 Let $\cK=\{\bw_j:j=1,2,\cdots,M\}$ be the sampling set whose elements are i.i.d. r.v.s
 on $\cM$ or $[0,1]^d$.
 
 Consider the equation (\ref{Fou}) and set $\bA=[e^{-\im 2\pi \bn\cdot\bw_j}]$.
 
 To do this, we use the compressed sensing technique, called
 the Basis Pursuit, which is the following $\ell_1$ minimization
 principle
 \beq
 \label{BP}
 \min_{C} \|C\|_1,\quad \hbox{s.t.}\,\, \bA C=(|F(\bw_j)|^2)^T 
 \eeq

\begin{theorem} \label{thm3}
Suppose 
\beq
\label{77}
{n\over \ln{n}}\geq C_0 \delta^{-2}s\ln^2{s} \ln{m} \ln{1\over \rho},\quad
\rho\in (0,1)
\eeq
 holds any $\delta<\sqrt{2}-1$.   Then 
  the Basis Pursuit  minimizer  $\hat X$ satisfies
 \beq
 \|\hat X-X\|_2&\leq & C_1s^{-1/2}\|X-X^{(s)}\|_1+C_2\ep
 \eeq
 for some constants $C_1$ and $C_2$. 
with probability at least $1-\rho$. Here $C_0, C_1$ and $C_2$ are
absolute constants. 
%In particular, every target vector of sparsity less than $S$
%can be exactly recovered  by  BP (\ref{L1}).
\end{theorem}
}

%\section{\red Magnitude retrieval}
%Travel time tomography, imaging by wavefront sensing

\commentout{
\begin{center}
   \begin{tabular}{l}
   \hline  
   \centerline{{\bf Error Reduction}}  \\ \hline
    At the $k$th iteration, start with $X_k$.\\
Step 1 :  Compute $Y_k = F \lambda X_k$ and let $Y_k(\bw) = |Y_k(\bw)| e^{i \phi_k(\bw)}$. \\
Step 2 :  Fit the intensity measurement by letting $Y'_k(\bw) = |Y(\bw)|e^{i \phi_k(\bw)}$. \\
Step 3 : Do inverse Fourier transform $X'_k = F^{-1} Y'_k$.\\
Step 4 : $X_{k+1} = \cP_o(X'_k)$, which is the orthogonal projection of $X'_k$ to the set of images satisfying\\
\qquad  \quad \quad  the object-domain constraints.\\
    \hline
   \end{tabular}
\end{center}
In Step $4$, one enforces the object-domain constraints on $X'_k$. Specifically, if $X(\bn)$ is known to be real,
\begin{equation*}
X_{k+1}(\bn) = \text{real}(X'_k(\bn)),
\end{equation*}
and if $X(\bn)$ is further known to be real positive
\begin{equation*}
X_{k+1}(\bn) = \left\{ \begin{array}{ll}
\text{real}(X'_k(\bn)) & \text{if } \text{real}(X'_k(\bn)) \ge 0 \\
0 & \text{if } \text{real}(X'_k(\bn)) < 0\end{array}. \right.
\end{equation*}
If $X(\bn)$ is complex with nonnegative real and imaginary parts,
$$\text{real}(X_{k+1}(\bn)) = \left\{ \begin{array}{ll}
\text{real}(X'_k(\bn)) & \text{if } \text{real}(X'_k(\bn)) \ge 0 \\
0 & \text{if } \text{real}(X'_k(\bn)) < 0\end{array}, \right.
$$
$$\text{imag}(X_{k+1}(\bn)) = \left\{ \begin{array}{ll}
\text{imag}(X'_k(\bn)) & \text{if } \text{imag}(X'_k(\bn)) \ge 0 \\
0 & \text{if } \text{imag}(X'_k(\bn)) < 0\end{array}. \right.
$$
Moreover, if the support $\Sigma$ of $X$ is given as a prior information,
$$X_{k+1}(\bn) = 0 \text{  if } \bn \notin S.$$
}

\section{Numerical examples}\label{sec:num}
\commentout{
\begin{figure}[!t]
  \centering
 \subfigure[]{
    \label{Bull} %% label for first subfigure
    \includegraphics[width = 2in]{figures/Bull.pdf}
    }
        \subfigure[]{
    \label{BullSample4} %% label for second subfigure
    \includegraphics[width = 2in]{figures/BullSample4.pdf}
    } 
    %%%%%%%%%%%%%%%%%%%%%%%%%%%%%%%%%%%%%%%%%%%%
%   \subfigure[]{
  %  \label{BullSample6} %% label for second subfigure
  %  \includegraphics[width = 1.5in]{figures/BullSample6.pdf}
   % }
      \subfigure[]{
    \label{BullRISample1dot1} %% label for second subfigure
    \includegraphics[width = 2in]{figures/BullRISample1dot1.pdf}
    }
  \caption{(a) The original object and reconstructions with
  (b) uniform  illumination, $\rho  = 4$ , relative error $  13\%$,
  relative Fourier magnitude residual $0.74\%$ and
  (c) random   phase illumination,  $\rho  = 0.55$, relative error $2.97\%$, relative Fourier magnitude residual $0.37\%$   (adapted from \cite{FL}).
  %(c) uniform  illumination, $\rho  = 6$, $\|\hat f- {f}_{\text{twin}}\|/\|f\|\approx  17\%$ (adapted from \cite{FL}).
  }
  \label{RealPositiveBull} %% label for entire figure
  \label{fig2}
\end{figure}
}

\commentout{
\begin{figure}[!t]
  \centering
 \subfigure[]{
    \label{Cameraman} %% label for first subfigure
    \includegraphics[width = 2in]{Asilomar11/figure/Cameraman.pdf}}
            \subfigure[]{
    \label{CameramanRISample1dot1} %% label for second subfigure
  \includegraphics[width = 2in]{Asilomar11/figure/CameramanRISample1dot1.pdf}}\\
     \subfigure[]{
    \label{CameramanSample4} %% label for second subfigure
    \includegraphics[width = 2in]{Asilomar11/figure/CameramanSample4.pdf}}
        \subfigure[]{
    \label{CameramanSample2} %% label for second subfigure
    \includegraphics[width =2in]{Asilomar11/figure/CameramanSample2.pdf}} 
    %%%%%%%%%%%%%%%%%%%%%%%%%%%%%%%%%%%%%%%
  \caption{(a) The original object and reconstructions with (b)  random   phase illumination,  $\rho=1.1$, relative error $\approx 2.61\%$; (c)   uniform  illumination, $\rho  = 4$,
 relative error $\approx  4.36\%$; (d) uniform  illumination,  $\rho = 2$ (adapted from \cite{FL}). }
  \label{Cameraman} %% label for entire figure
\end{figure}
}

\begin{figure}[hthp]
  \centering
  \subfigure[]{
         \includegraphics[width = 5cm]{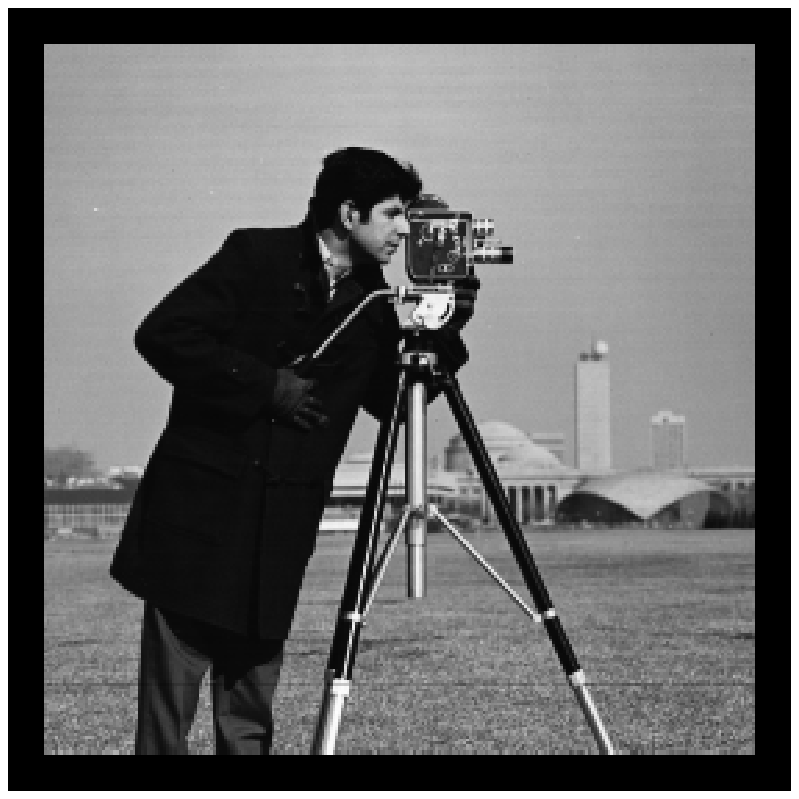}}
         \hspace{1.3cm}
             %%%%%%%%%%%%%%%%%%%%%%%%%%%%%%%%%
  \subfigure[]{
         \includegraphics[width = 5cm]{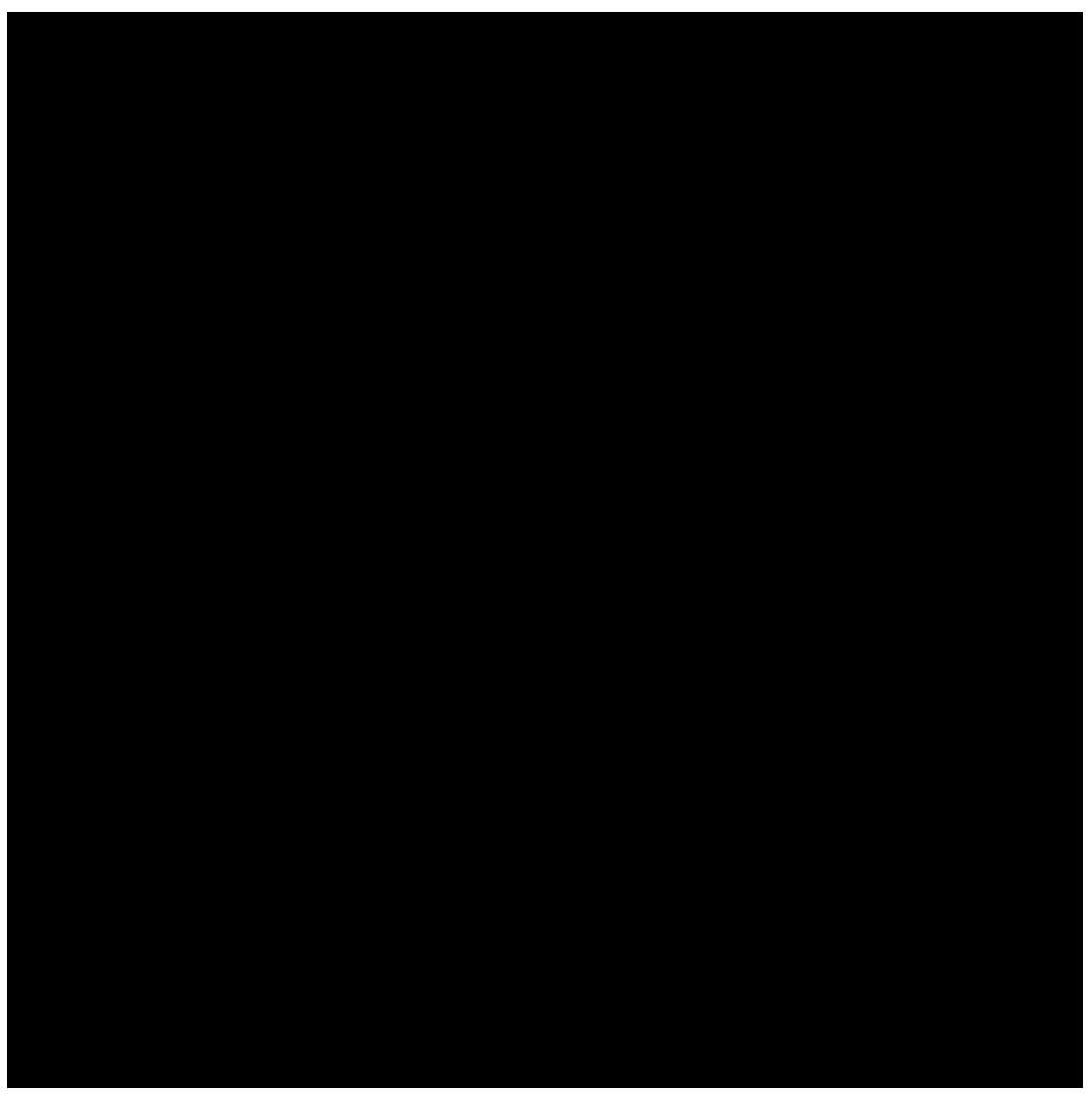}}\\
             %%%%%%%%%%%%%%%%%%%%%%%%%%%%%%%%%
  \subfigure[]{
    \includegraphics[width = 6.5cm]{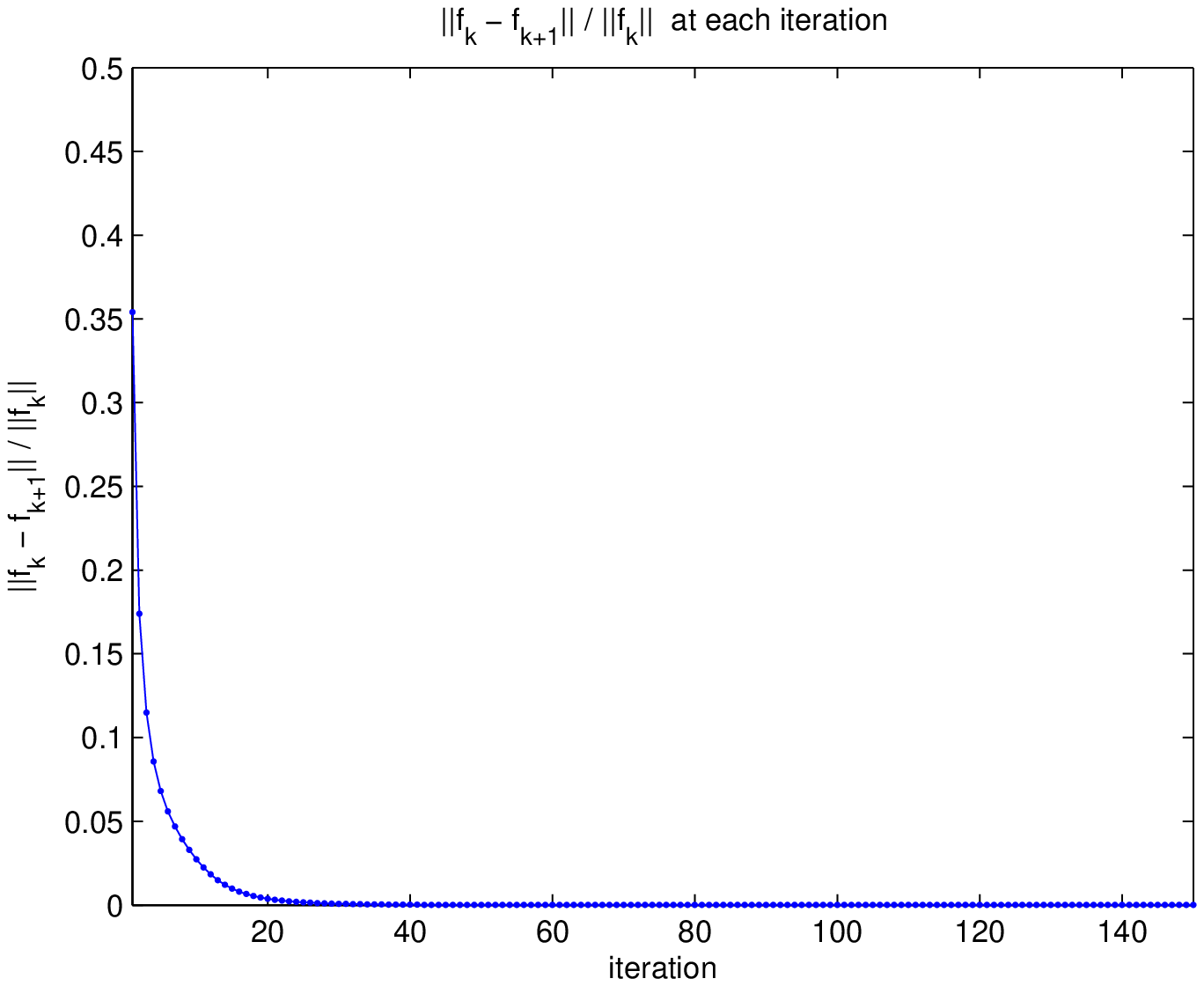}}
        %%%%%%%%%%%%%%%%%%%%%%%%%%%%
          \subfigure[]{
         \includegraphics[width = 6.5cm]{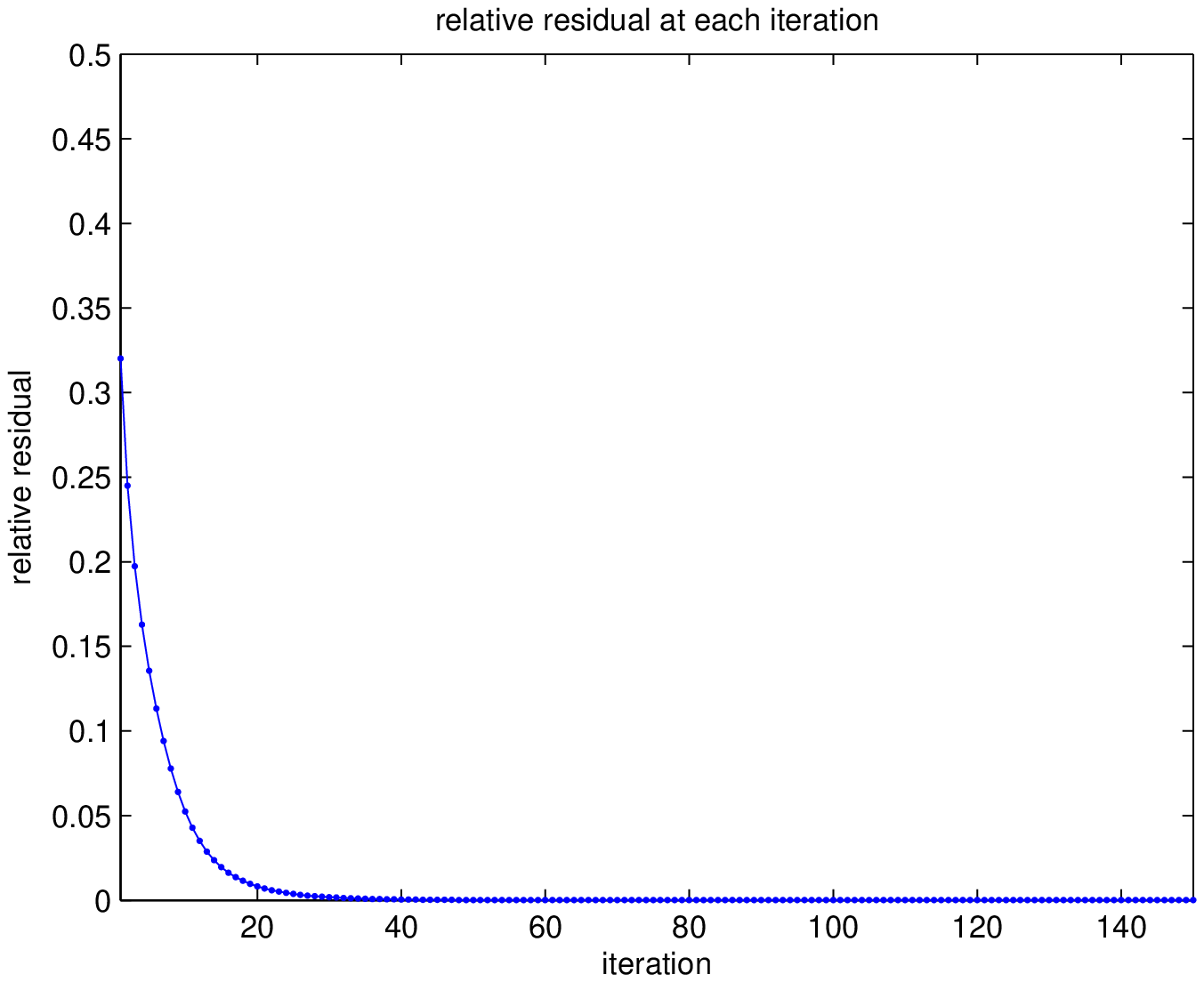}}
            % \subfigure[]{
        % \includegraphics[width = 6.5cm]{ERstudy2/Cameraman1PhaseIn1by10PaddingSample4Noise0ERErr.eps}}
   
             %%%%%%%%%%%%%%%%%%%%%%%%%%%%%%%%%
     \caption{ ER reconstruction with random phase illumination: (a) recovered image (b) difference between the true and recovered images  (c) relative change $\|f_{k+1}-f_k\|/\|f_k\|$  (e) relative residual $\||\widetilde F|-|\bPhi\Lambda f_k|\|/\|\widetilde F\|$ versus number of iterations.}\label{fig2}\end{figure}

\begin{figure}[hthp]
  \centering
  \subfigure[]{
         \includegraphics[width = 5cm]{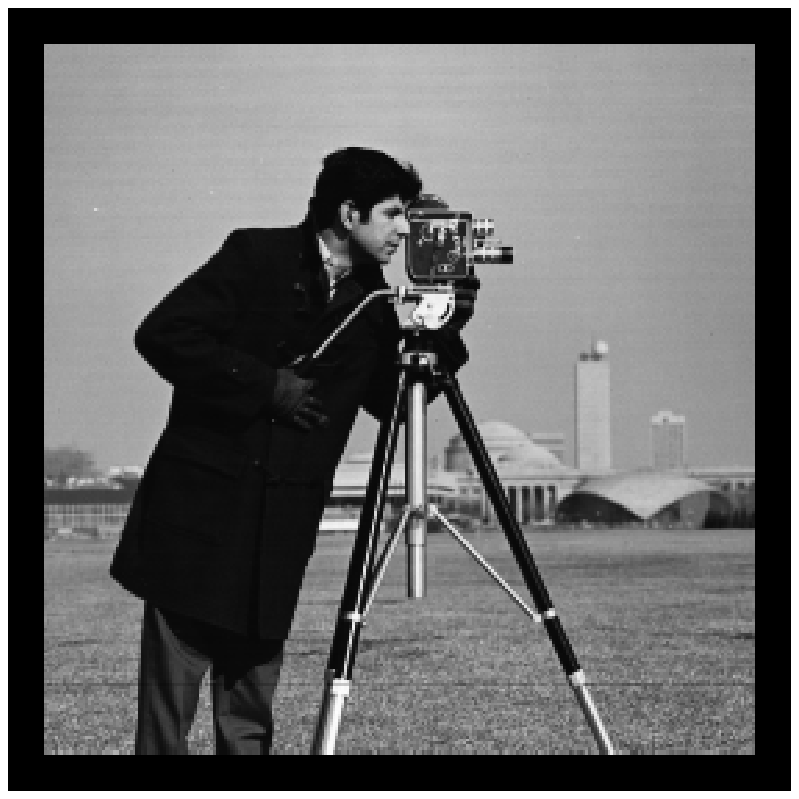}}\hspace{1.3cm}
             %%%%%%%%%%%%%%%%%%%%%%%%%%%%%%%%%
\subfigure[]{
     \includegraphics[width = 5cm]{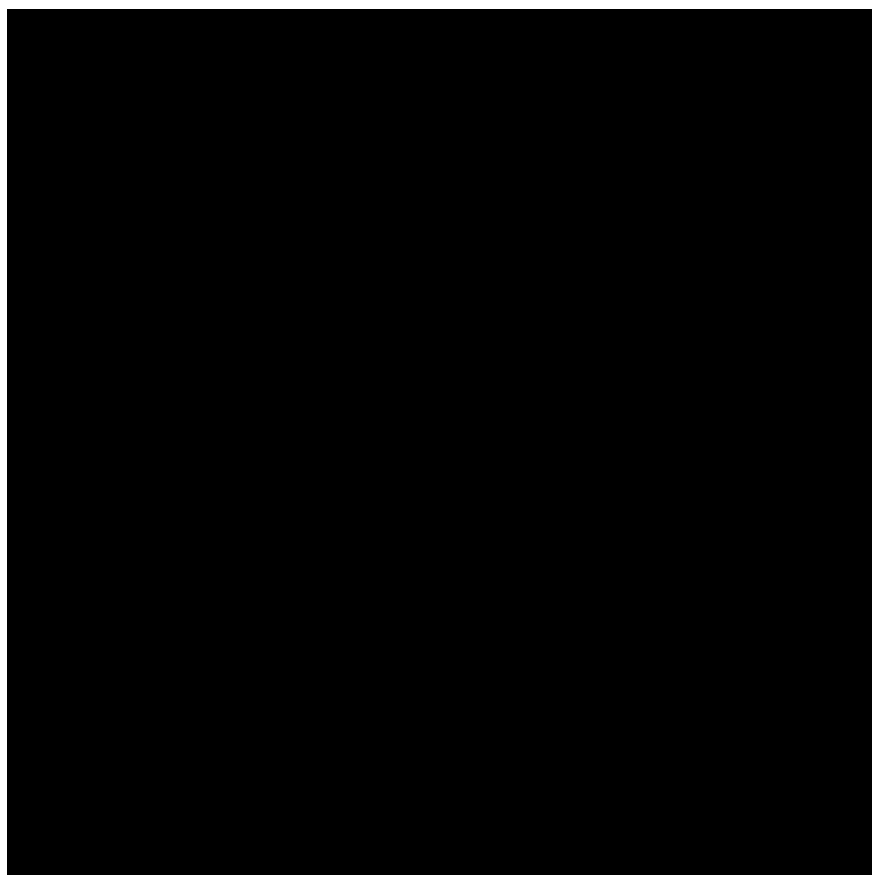}}\\
             %%%%%%%%%%%%%%%%%%%%%%%%%%%%%%%%%
  \subfigure[]{
    \includegraphics[width = 6.5cm]{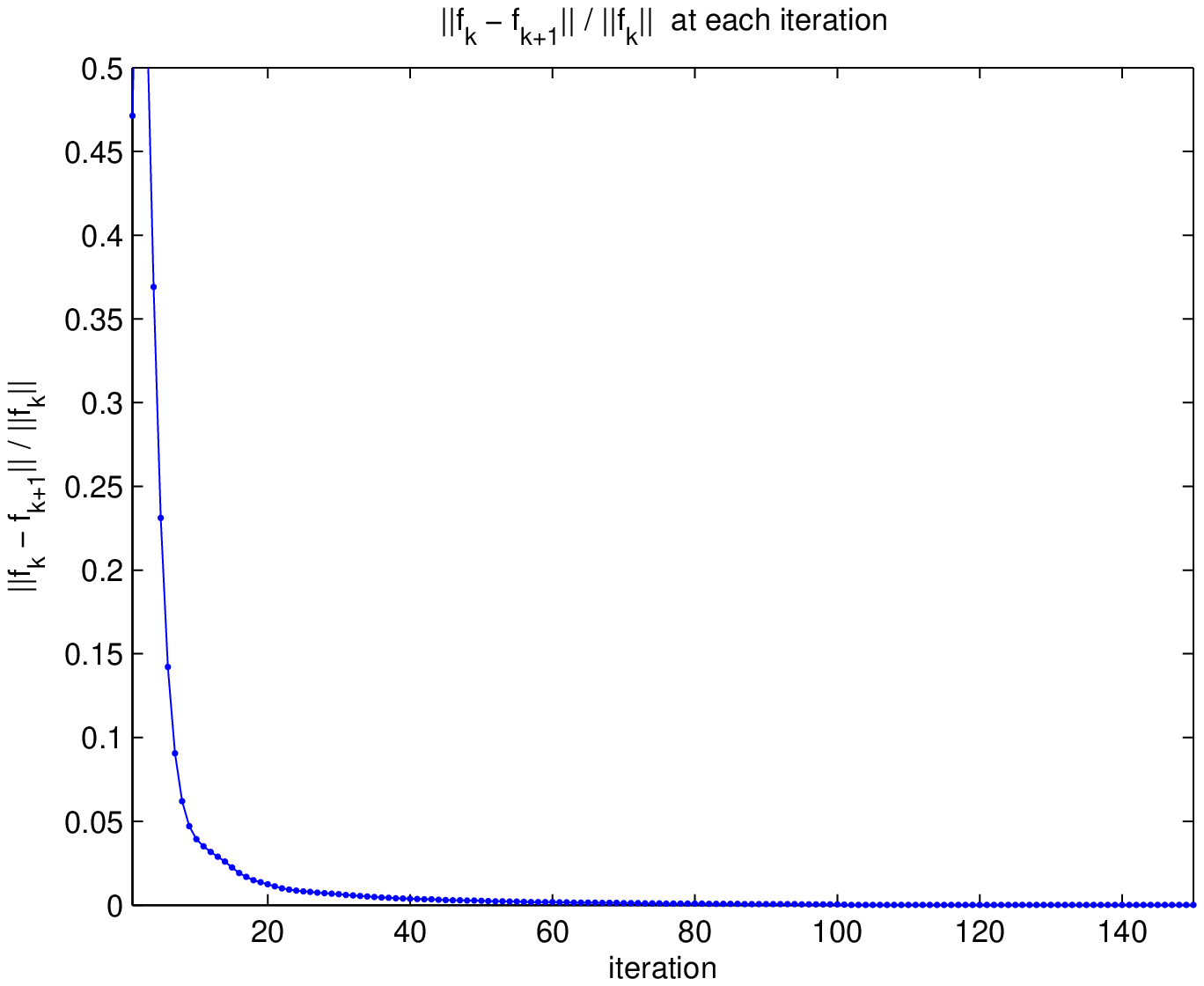}}
         %%%%%%%%%%%%%%%%%%%%%%%%%%%%%%%%%
          \subfigure[]{
         \includegraphics[width = 6.5cm]{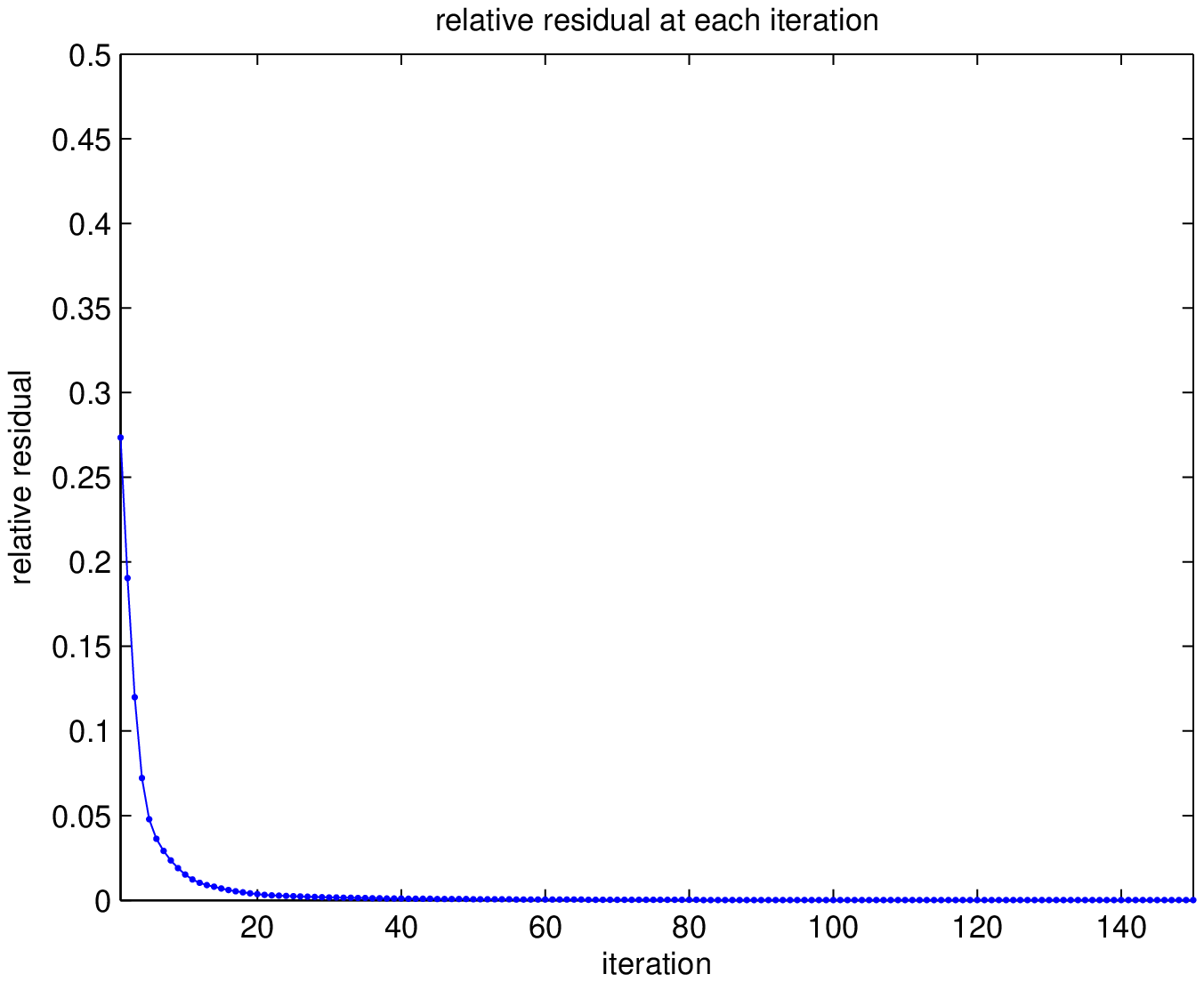}}
             %%%%%%%%%%%%%%%%%%%%%%%%%%%%%%%%%
                       %   \subfigure[]{
       %  \includegraphics[width = 6.5cm]{ERstudy2/Cameraman1PhaseIn1by10PaddingSample4Noise0HIOErr.eps}}

     \caption{HIO reconstruction with random phase illumination: (a) recovered image (b) difference between the true and recovered images (c) relative change  (e) relative residual versus number of iterations. The final 50 iterations are ER.}
     \label{fig4}
     \end{figure}

\begin{figure}[hthp]
  \centering
             \subfigure[]{
         \includegraphics[width = 6.5cm]{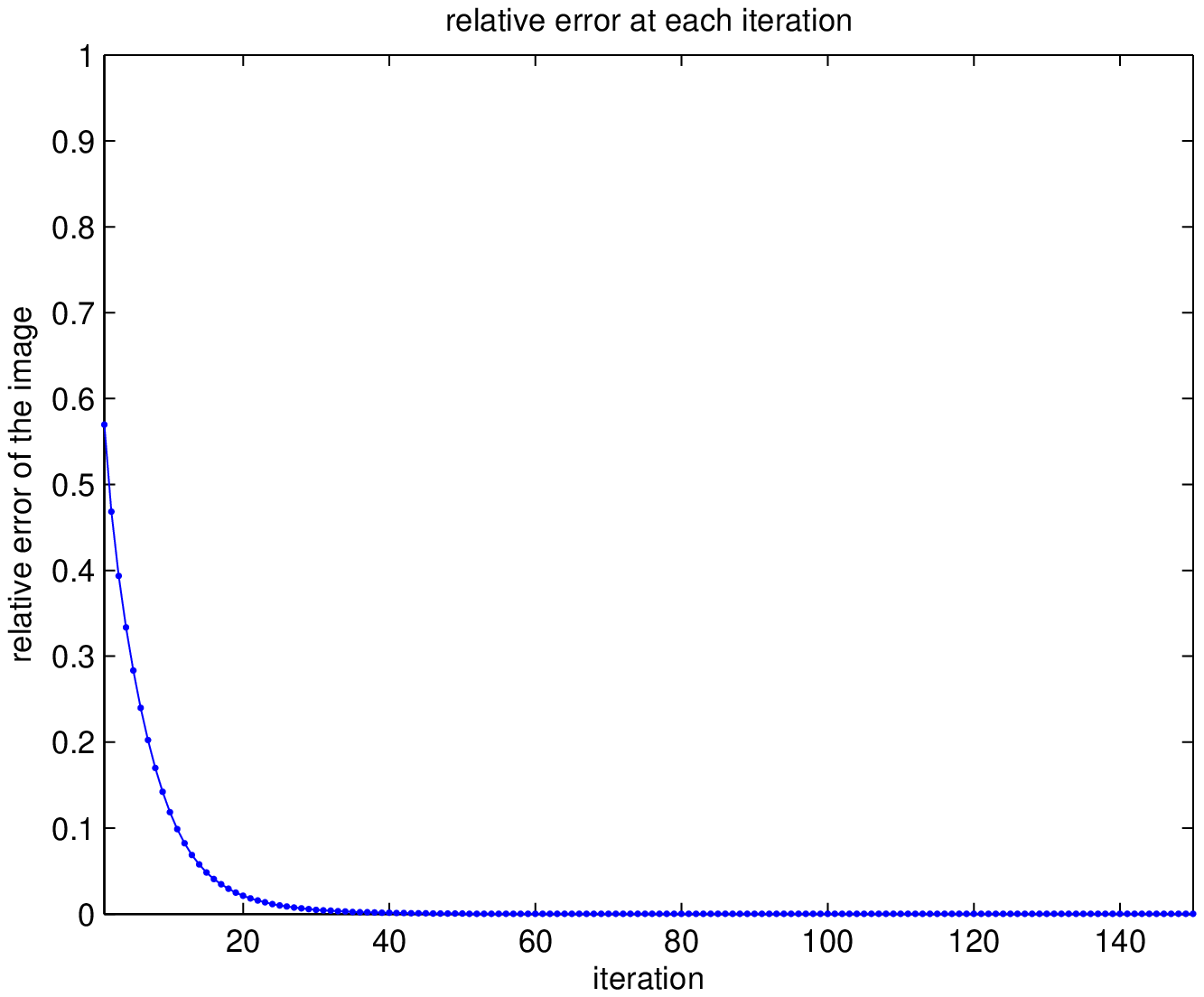}}
                    \subfigure[]{
       \includegraphics[width = 6.5cm]{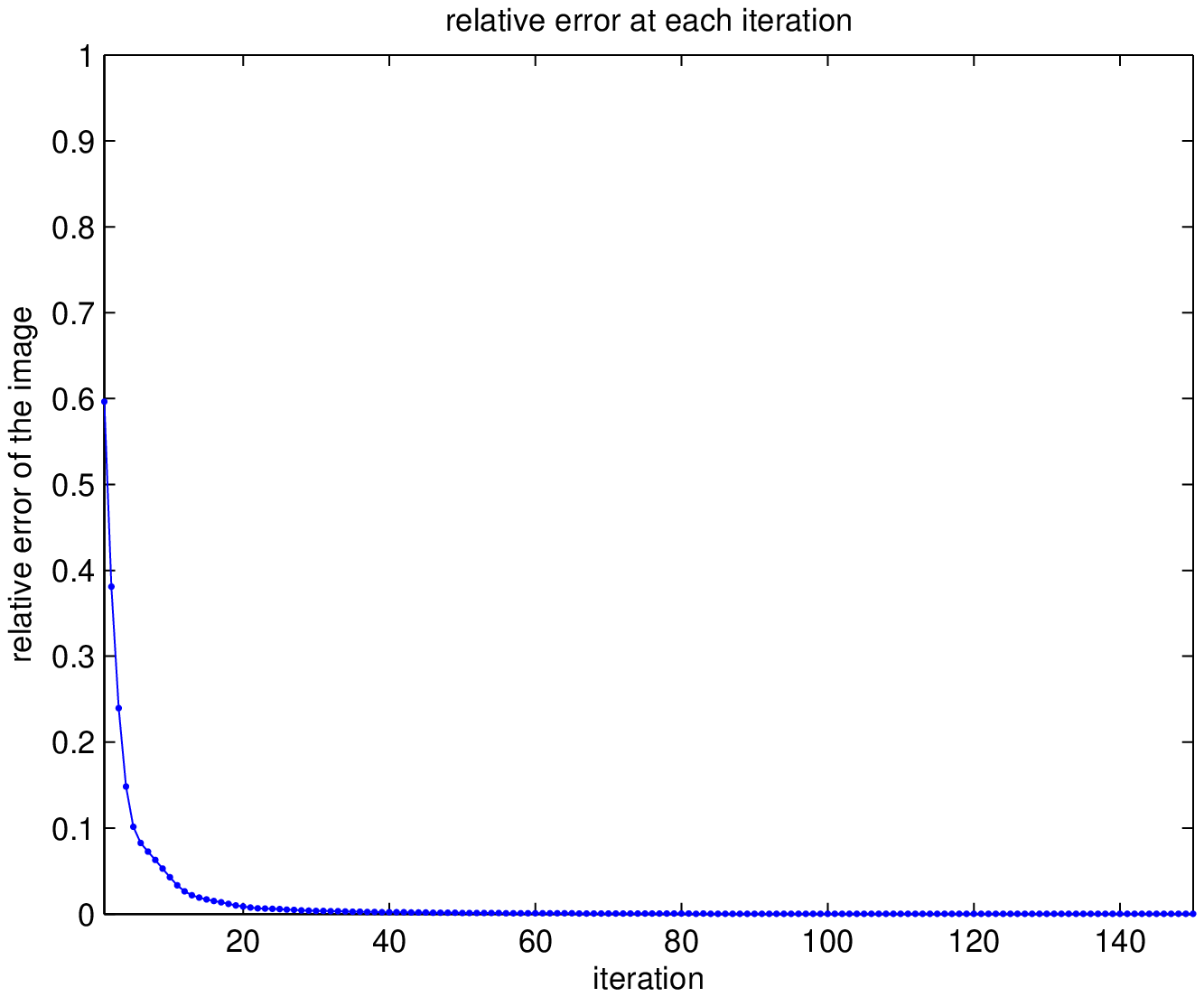}}
     \caption{ Relative error $\|f_k-f\|/\|f\|$  with (a) ER and (b) HIO versus number of iterations.}
     \label{fig6}\end{figure}

\begin{figure}[hthp]
  \centering
  \subfigure[]{
         \includegraphics[width = 5cm]{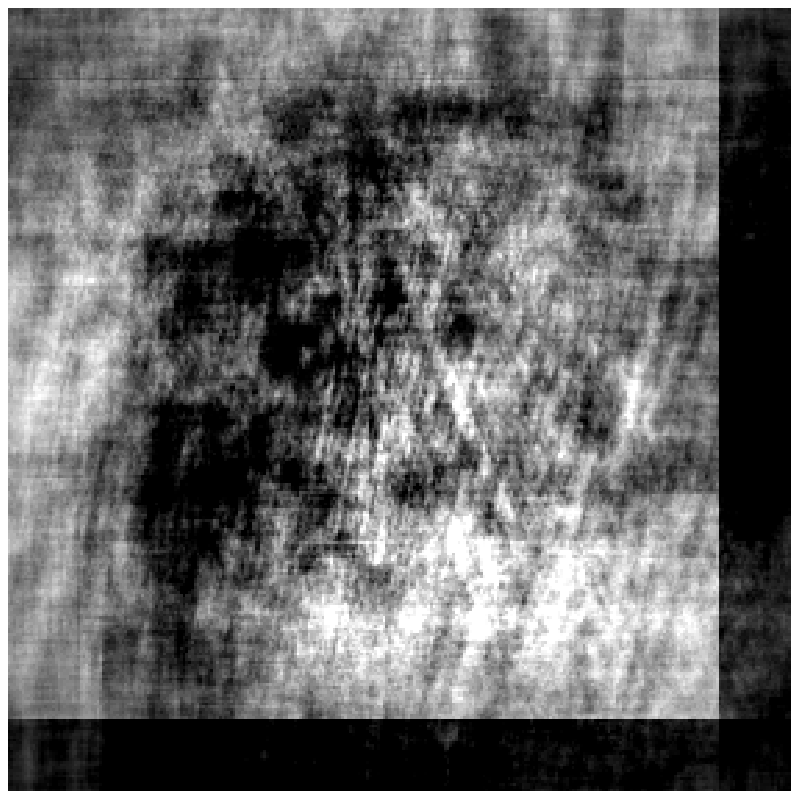}}\hspace{1.3cm} 
             %%%%%%%%%%%%%%%%%%%%%%%%%%%%%%%%%
\subfigure[]{
       \includegraphics[width = 5cm]{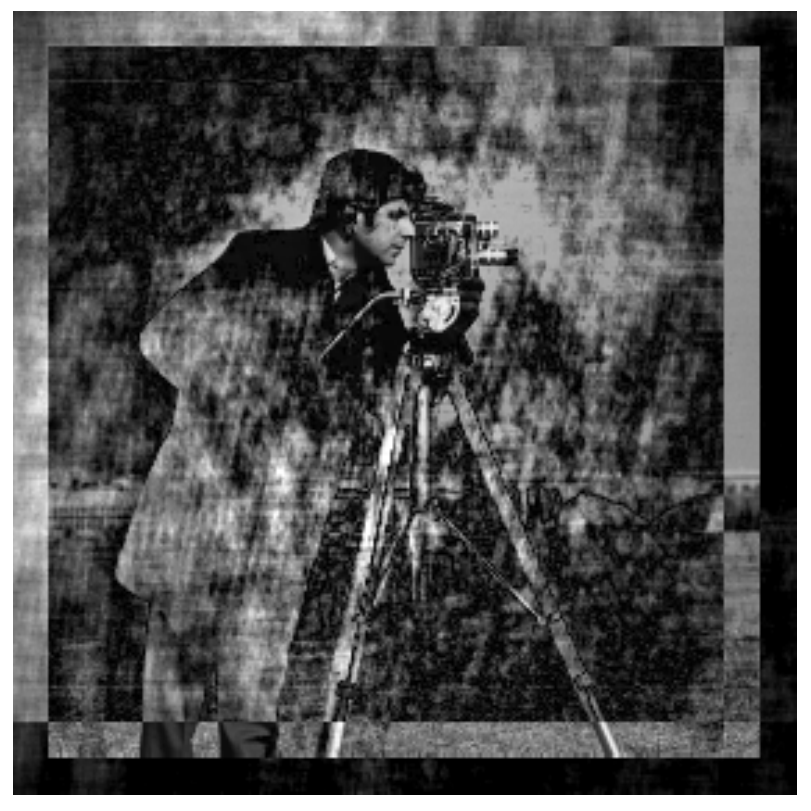}}\\
             %%%%%%%%%%%%%%%%%%%%%%%%%%%%%%%%%
  \subfigure[]{
    \includegraphics[width = 6.5cm]{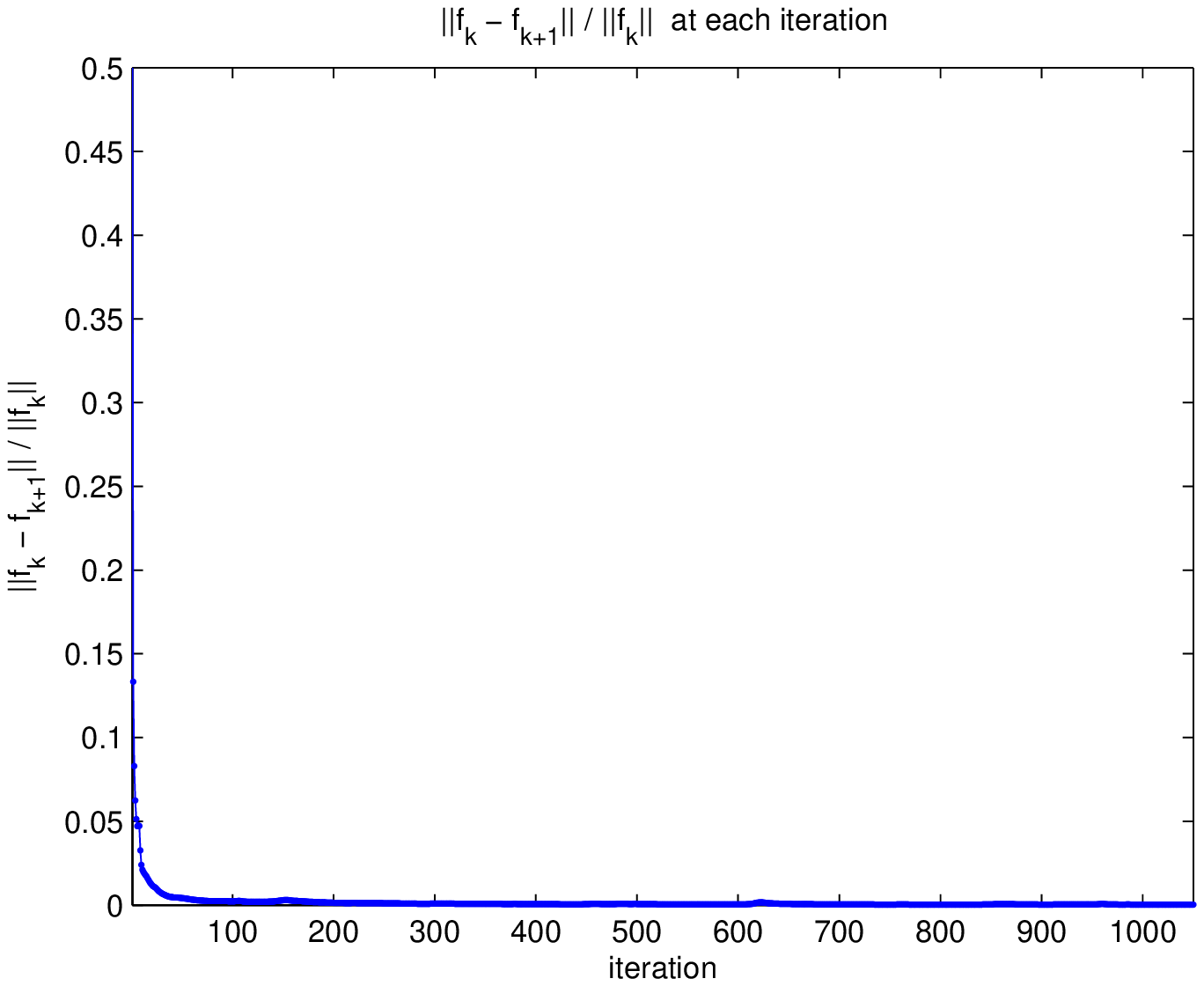}}
        %%%%%%%%%%%%%%%%%%%%%%%%%%%%%%%%%
          \subfigure[]{
         \includegraphics[width = 6.5cm]{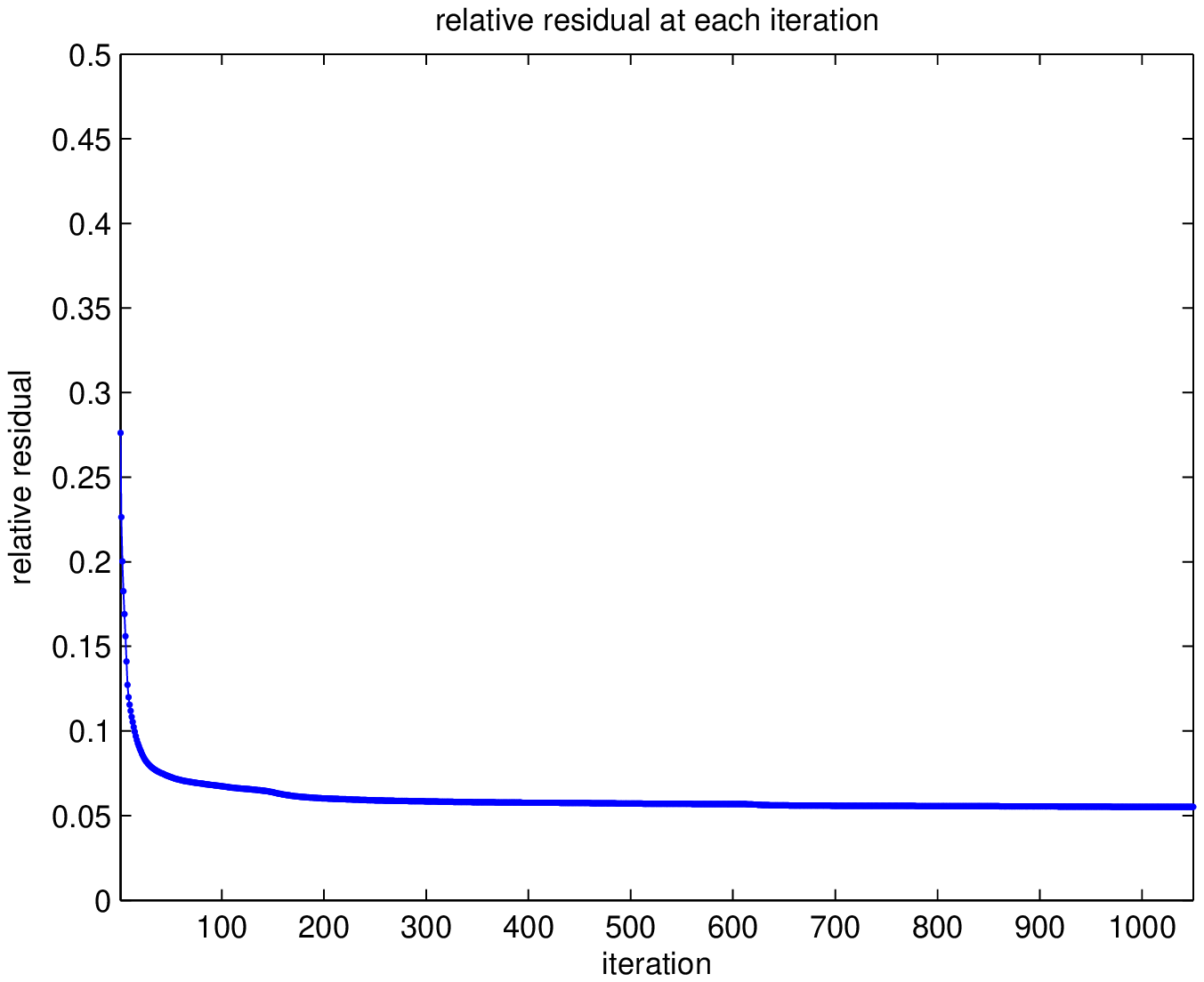}}
             %%%%%%%%%%%%%%%%%%%%%%%%%%%%%%%%%
                        %  \subfigure[]{
        % \includegraphics[width = 6.5cm]{ERstudy2/Cameraman1Uniform0PaddingSample4Noise0ERErr.eps}}
     \caption{ER reconstruction with uniform illumination:  (a) recovered image (b) difference between the true and recovered images (c) relative change (d) relative residual versus number of iterations.}
     \label{fig3}
     \vspace{0cm}
\end{figure}

\begin{figure}[hthp]
  \centering
  \subfigure[]{
         \includegraphics[width = 5cm]{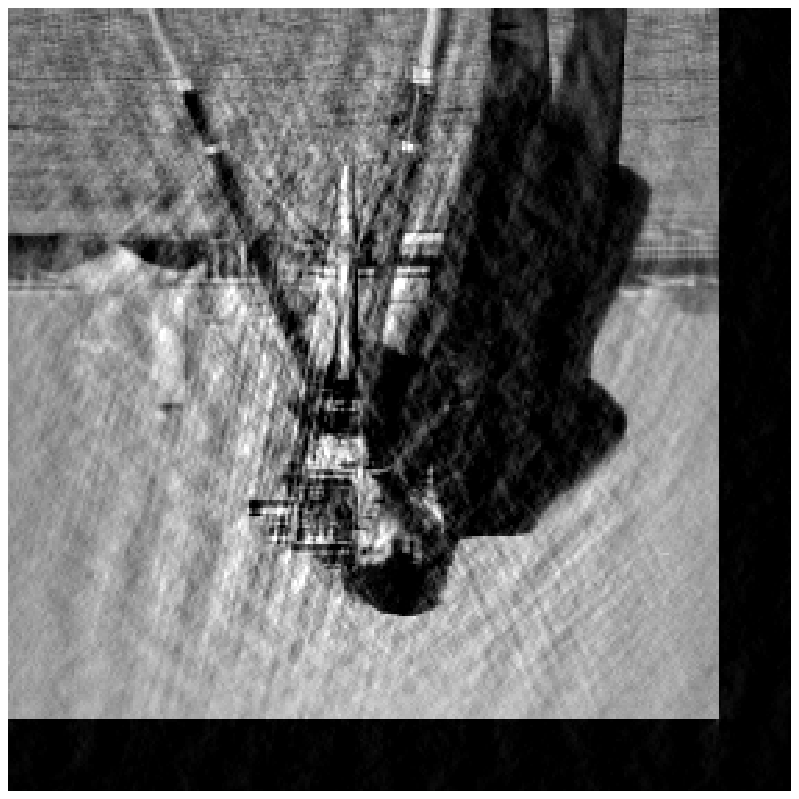}}\hspace{1.3cm}
             %%%%%%%%%%%%%%%%%%%%%%%%%%%%%%%%%
  \subfigure[]{
         \includegraphics[width = 5cm]{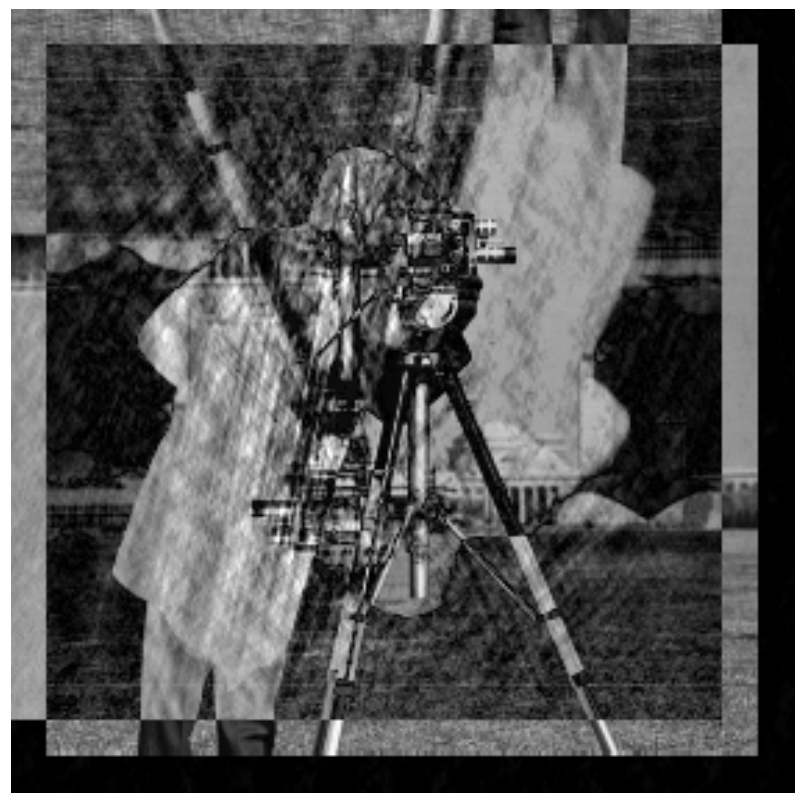}}\\
             %%%%%%%%%%%%%%%%%%%%%%%%%%%%%%%%%
  \subfigure[]{
    \includegraphics[width = 6.5cm]{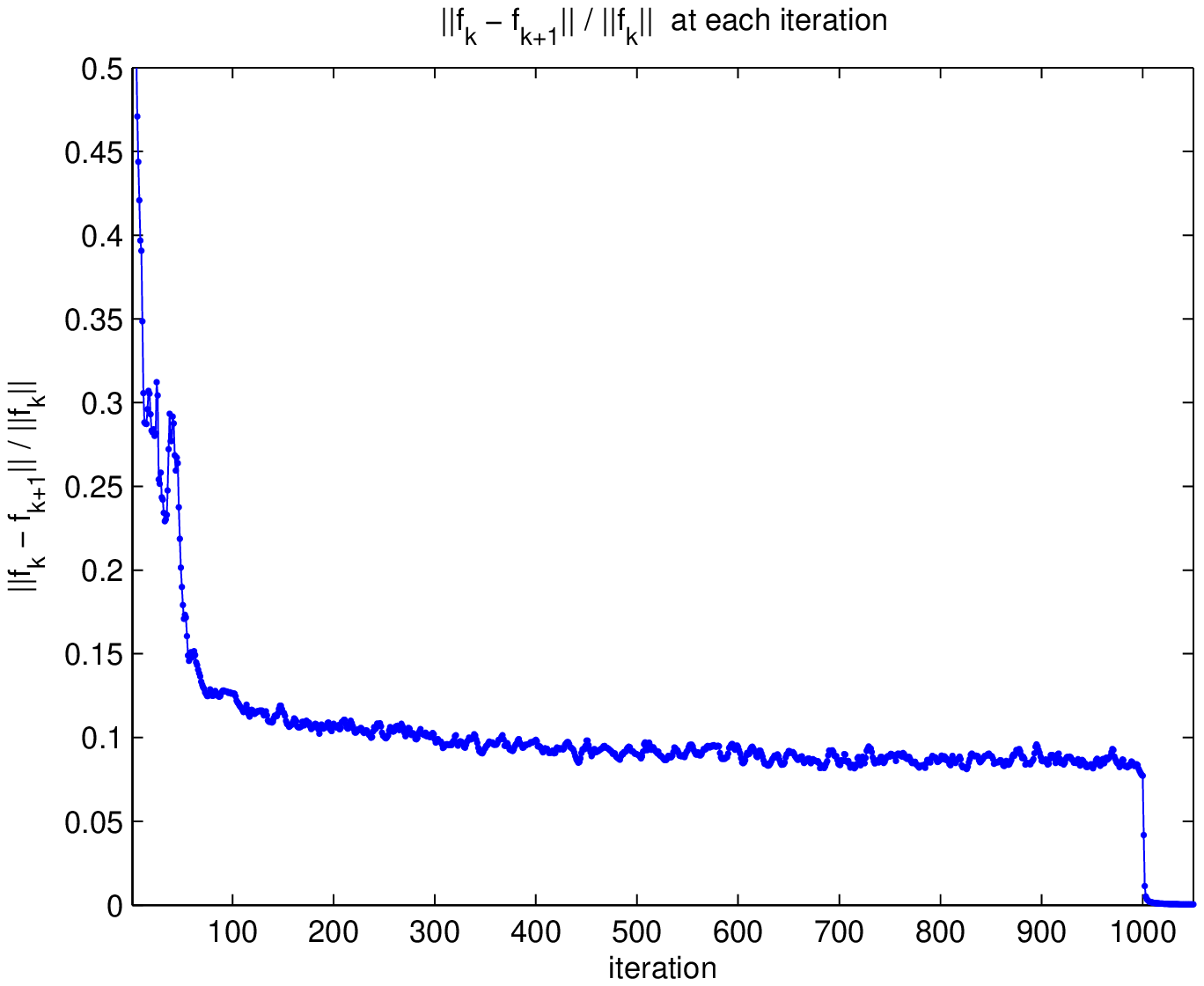}}
        %%%%%%%%%%%%%%%%%%%%%%%%%%%%%%%%%
          \subfigure[]{
         \includegraphics[width = 6.5cm]{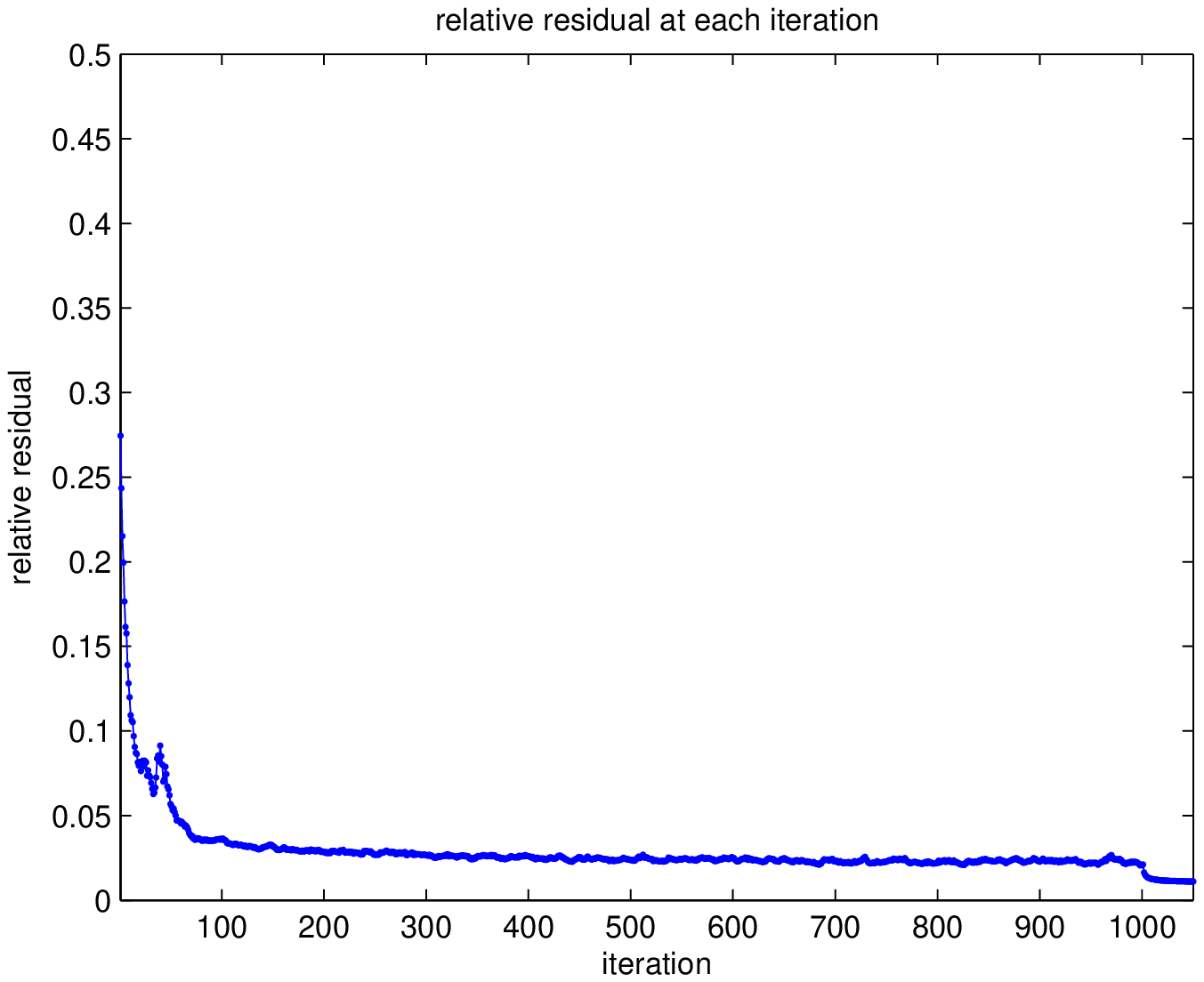}}
             %%%%%%%%%%%%%%%%%%%%%%%%%%%%%%%%%
                      %    \subfigure[]{
        % \includegraphics[width = 6.5cm]{ERstudy2/Cameraman1Uniform0PaddingSample4Noise0HIOErr.eps}}
     \caption{HIO reconstruction: (a) recovered image (b) difference between the true and recovered images (c)  relative change (d) relative residual versus number of iterations. The final 50 iterations are ER,  causing a dip in the relative change and residual.}
     \label{fig5}
\end{figure}

\begin{figure}[hthp]
  \centering
                       \subfigure[]{
      \includegraphics[width = 6.5cm]{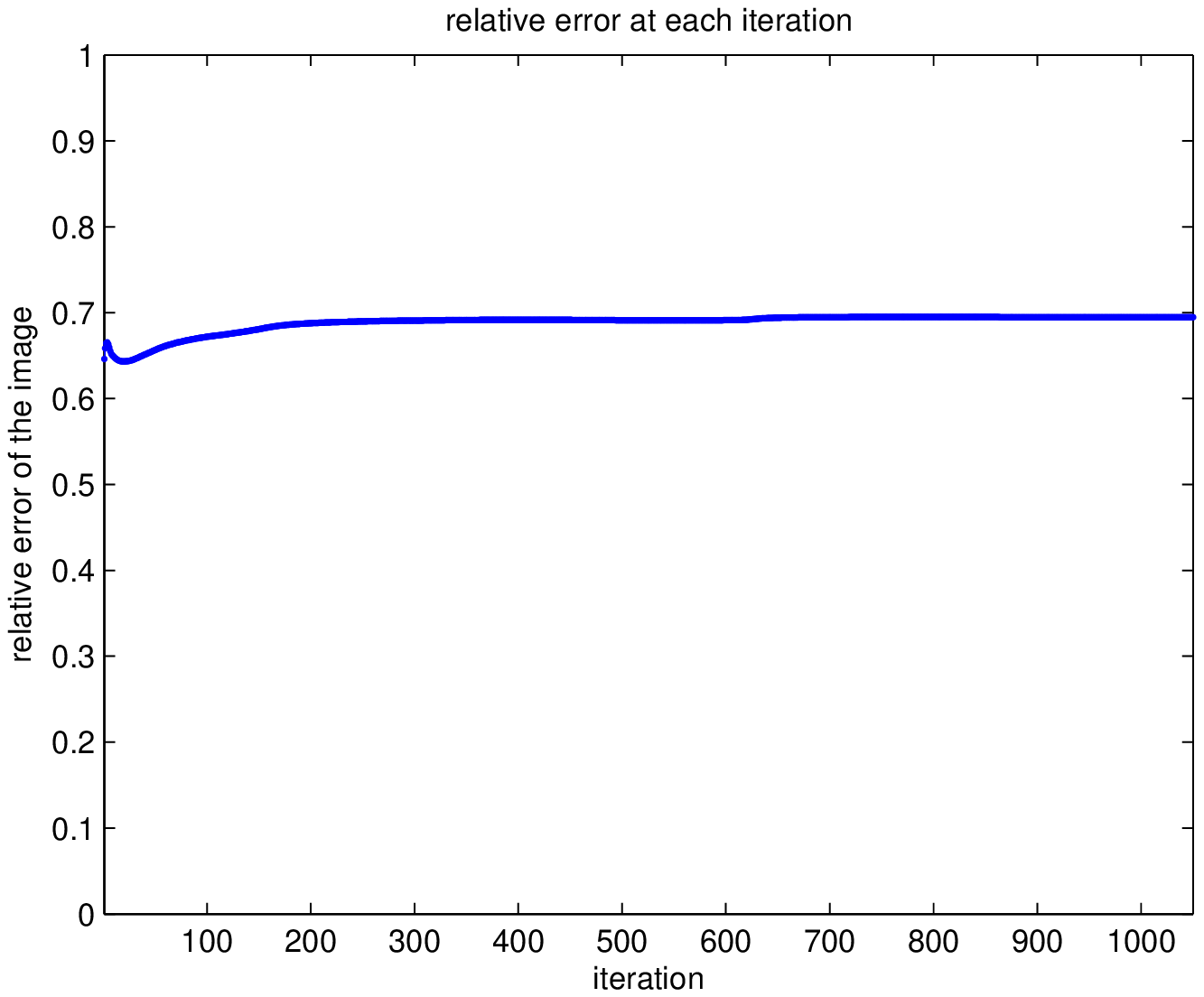}}
                           \subfigure[]{
        \includegraphics[width = 6.5cm]{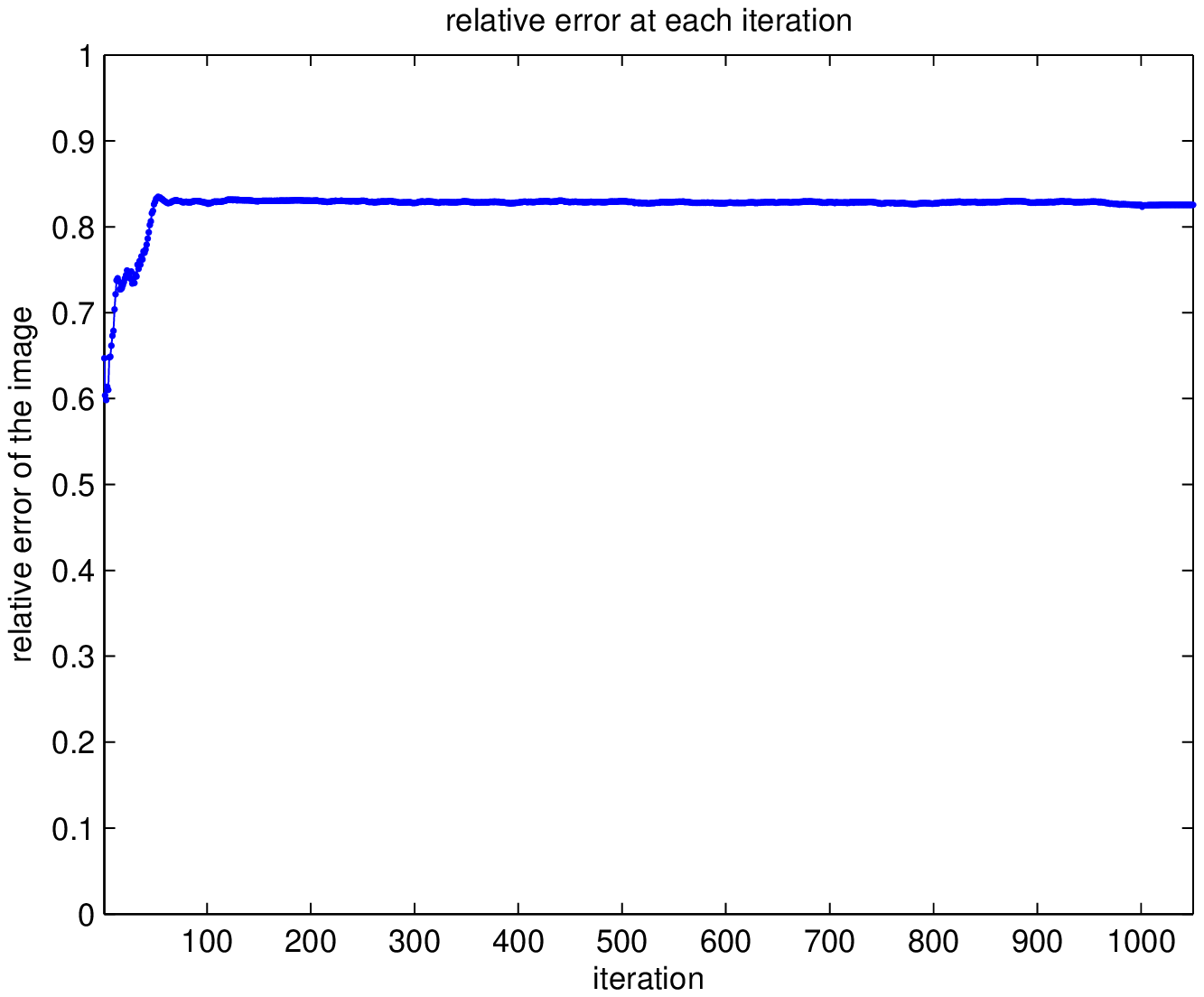}}
     \caption{Relative error with (a) ER and (b) HIO versus number of iterations.}
     \label{fig7}
\end{figure}

 Our previous
numerical study \cite{FL} and 
the following numerical examples give a glimpse  of how  the quality and efficiency of reconstruction  can be improved by random   illumination.  
%Random phase illumination can be
%implemented by positioning a phase diffuser close to the object. 

We test the case of random {\em phase} illumination
on a real, positive $269\times 269$ image consisting of the original $256\times 256$  Cameraman in the middle,  surrounded by a black margin (zero padding) of 13 pixels in width (Figure \ref{fig2}(a)).
We synthesize and sample the Fourier magnitudes at the Nyquist rate (Remark \ref{rmk1}) and  implement the standard Error Reduction (ER) and Hybrid-Input-Output (HIO)   algorithms   
%followed by $10$ steps of error-reductions \cite{Fie}
in the framework of  the oversampling
method \cite{MS, MSC}. 
By Corollary \ref{cor}, 
absolute uniqueness holds with a random phase illumination. 
 
\commentout{
To test how many Fourier magnitude data are needed for phasing, we define the sampling rate 
\[
\rho=\# \hbox{ Fourier magnitude data}/\# \hbox{ image pixels}.
\]
%\[
%\rho=\# \,\,\hbox{(independent) Fourier magnitude data}/\hbox{Image's (real) degrees of freedom}
%\]
%which is different from the definition given in \cite{FL}. 
%The real degrees of freedom are
%$2|\cN|$ for a complex-valued image and
%$|\cN|$ for a real-valued image. 
The uniqueness results above are established for  $\rho=4$ (in two dimensions). 
}
\commentout{Notice that in the case
of a real-valued object (cf. Corollary \ref{cor}) the
sampling grid $\cM$ can be reduced by half 
because of the symmetry $F^*(1-\bw)=F(\bw),\forall \bw\in \cM$.} 
%We implement the standard Error Reduction (ER) and Hybrid-Input-Output (HIO)   algorithms   
%followed by $10$ steps of error-reductions \cite{Fie}
%in the framework of  the oversampling
%method \cite{MS, MSC}. 
%which converts the Fourier magnitude data with $\rho>1$ into a support constraint to reduces the ambiguity of spatial shift (but not the twin image). 

Let $\bPhi$ and $\Lambda$ be the Fourier transform
and the diagonal matrix $\hbox{\rm diag}{[\lambda(\bn)]}$  representing the illumination.
For the uniform illumination $\Lambda=\bI$. 
The ER and HIO algorithms are described below.
\vspace{-0cm}
\begin{center}
   \begin{tabular}{l}
   \hline  
   \centerline{{\bf ER algorithm}}  \\ \hline
 Input: Fourier magnitude data $\{\widetilde F(\bw)\}$, initial guess $f_0$.\\   
  %While $\|f_{k+1}-f_k\|/\|f_k\|\geq \ep, k=0,1,2,\cdots$,\\
  Iterations:\\
\quad $\bullet$  Update Fourier phase: $G_k= \bPhi \Lambda f_k= |G_k(\bw)| e^{i \theta_k(\bw)}$. \\
\quad $\bullet$  Update Fourier magnitude:  $G'_k(\bw) = |\widetilde F(\bw)|e^{i \theta_k(\bw)}$. \\
\quad $\bullet$ Impose object constraint $f_{k+1}(\bn) = \left \{ 
\begin{array}{ll} f'_k(\bn) &   \text{  if $f'_k(\bn)=\Lambda^{-1}\bPhi^* G'_k(\bn)\geq 0$} \\
0 & \text{ otherwise}
\end{array} 
\right. $\\
%Otherwise, output $f_{k+1}$.\\
    \hline
   \end{tabular}
\end{center}

%\vspace{1cm}

In the original version of HIO \cite{Fie}, the hard thresholding is replaced by 
\[
f_{k+1}(\bn) = \left \{ 
\begin{array}{cl} f'_k(\bn) &   \text{if $f'_k(\bn)=\Lambda^{-1}\bPhi^* G'_k(\bn)\geq 0$} \\
 f_k(\bn)-\beta f'_k(\bn) & \text{otherwise}
\end{array} 
\right. 
\]
where the feedback parameter $\beta=0.9$ is used in the simulations. ER has the desirable property that the residual 
$\||\widetilde F|-|G_k|\|$ is reduced after each iteration under either 
uniform or random phase illumination \cite{Fie,FL2}. When absolute uniqueness holds, a vanishing residual then implies a vanishing reconstruction error. 
%But we do not yet have the guarantee of vanishing residual for  ER or HIO.
\commentout{
We use $\ep=0.01\%$ in the stopping rule for ER.  For HIO + ER, HIO is stopped with $\ep=1\%$ with a maximal $1000$ iterations and ER is terminated when $\|f_{k+1}-f_k\|/\|f_k\|<0.01\%$.
}

\commentout{
\begin{figure}[!t]
  \centering
 \subfigure[]{
    \label{Phantom} %% label for first subfigure
    \includegraphics[width = 2in]{figures/Phantom.pdf}}
         \subfigure[]{
    \label{PhantomSample4} %% label for second subfigure
    \includegraphics[width = 2in]{figures/PhantomSample4.pdf}}
    %    \subfigure[]{
   % \label{CameramanSample2} %% label for second subfigure
   % \includegraphics[width =1.5in]{figures/CameramanSample2.pdf}} 
            \subfigure[]{
    \label{PhantomRISample1dot1} %% label for second subfigure
  \includegraphics[width = 2in]{figures/PhantomRISample1dot1.pdf}}
    %%%%%%%%%%%%%%%%%%%%%%%%%%%%%%%%%%%%%%%
  \caption{(a) The original object and reconstructions with (b)   uniform   illumination, $\rho  = 4$,
 relative error $ 115.1\%$, relative Fourier magnitude residual  $4.3\%$  and (c)  random   phase illumination,  $\rho=1.1$, relative error $0.17\%$, relative Fourier magnitude residual  $0.04\%$ (adapted from \cite{FL2}).
 % (d) constant  illumination,  $\rho = 2$.
  }
  \label{RealPositivePhantom} %% label for entire figure
  \label{fig1}
\end{figure}
}

Figure \ref{fig2} shows the results of ER reconstruction
with random phase illumination. The ER iteration converges to the true
image after 40 iterations. 
 For HIO reconstruciton  we apply 50 ER iterations
after 100 HIO iterations as suggested in \cite{Mar2}. 
HIO has  essentially the same
performance as  ER (Figure \ref{fig4} (a), (b)). 
The relative residual curve, Figure \ref{fig4}(d), and
the relative error curve, Figure \ref{fig6}, however,
indicate a small improvement by HIO.
The close proximity between the vanishing residual
curve and the vanishing error curve for ER and HIO
reflects the absolute uniqueness under random illumination.

With uniform illumination, ER produces a poor result 
(Figure \ref{fig3} (a)), resulting a $72.8\%$ error (Figure \ref{fig3} (b)) after more than 1000 iterations. The relative change
curve, Figure \ref{fig3}(c), indicates stagnation or convergence to a fixed point after 100 iterations and the relative residual plot, Figure \ref{fig3}(d),
shows non-convergence to the true image.  For HIO reconstruction,  
 we augment it with
50 ER iterations at the end of  1000 HIO iterations. While
HIO improves the performance of ER but still leads to a shifted, inverted image which is
 also severely distorted  (Figure \ref{fig5} (a), (b)).  The ripples and stripes in Figure \ref{fig5} (a) are a well known
 artifact of HIO reconstruction \cite{Fie, FW}.
 As expected,  HIO reduces the residual and does not
stagnate as much as ER (Figure \ref{fig5} (c), (d)) but
 its error is greater than that of ER 
due  to the interferences from shifted and twin images 
 present under the uniform illumination (Figure \ref{fig7}).

\commentout{
With {\em single} random phase illumination and the sampling ratio $\rho=1.1$, mere 100 HIO iterations 
are sufficient to achieve accurate recovery (Figure \ref{Cameraman} (b)).
In contrast, with uniform  illumination and
the sampling ratio $\rho = 4$,  $3000$ HIO iterations do not deliver a
comparable recovery. The recovered  image is contaminated with an artificial pattern of stripes (Figure \ref{Cameraman} (c)). 
When the sampling ratio is reduced to $\rho=2$ in the case
of uniform illumination, the HIO reconstruction 
is catastrophic, partly due to the interference of the
twin image (Figure \ref{Cameraman} (d)).
}

\commentout{
Let $\cN$ be the square (in the case of the phantom) or rectangular (in the case of Picasso's  bull) frame of image. A tight  image  has a tightly  defined support (i.e. $\cN$)
while a loose  image has a loosely defined support
(i.e. a proper subset of $\cN$). Objects with a loosely defined  support is more challenging to reconstruct. 

In both cases, the reconstruction error is high ($13\%$ for the bull and $115.1\%$ for the phantom) but the residual is
low ($0.37\%$ for the bull and $0.04\%$ for the phantom)
indicating the iterative process has more or less converged.
Hence the reconstruction error should be attributed to
the lack of uniqueness  
rather than the lack of convexity of phasing with uniform illumination. 

In principle, the stagnation problem (large number of iterations) may be  due to the lack of  convexity or
uniqueness. But consider Figure \ref{fig2}(c) and
 Figure \ref{fig1}(c): 
With just a { single} random phase illumination,
both problems with stagnation and error disappear and 
phasing with {\em  100 HIO iterations 
and  $\rho=1.1$}  
achieves  accurate, high-quality  recovery. 
}

To summarize, under a random phase illumination,
the problems of stagnation and error disappear and 
phasing with ER/HIO achieves  accurate, high-quality  recovery. 
These experiments confirm our belief that a central barrier to
stable and accurate phasing by the standard methods is the lack of absolute uniqueness.

\commentout{
\begin{figure}[!t]
  \centering
 \subfigure[]{
    \label{Cameraman} %% label for first subfigure
    \includegraphics[width = 2in]{figures/Cameraman.pdf}}
         \subfigure[]{
    \label{CameramanSample4} %% label for second subfigure
    \includegraphics[width = 2in]{figures/CameramanSample4.pdf}}
    %    \subfigure[]{
   % \label{CameramanSample2} %% label for second subfigure
   % \includegraphics[width =1.5in]{figures/CameramanSample2.pdf}} 
            \subfigure[]{
    \label{CameramanRISample1dot1} %% label for second subfigure
  \includegraphics[width = 2in]{figures/CameramanRISample1dot1.pdf}}
    %%%%%%%%%%%%%%%%%%%%%%%%%%%%%%%%%%%%%%%
  \caption{(a) The original object and reconstructions with (b)   uniform   illumination, $\rho  = 4$,
 relative error $\approx  4.36\%$ and (c)  random   phase illumination,  $\rho=1.1$, relative error $\approx 2.61\%$  (adapted from \cite{FL}).
 % (d) constant  illumination,  $\rho = 2$.
  }
  \label{RealPositiveCameraman} %% label for entire figure
  \label{fig1}
\end{figure}
}

\section{Conclusions}\label{sec5}
In conclusion, we have proposed random illumination to
address the uniqueness problem of  phase retrieval.  
For general random illumination we have  proved
almost sure irreducibility for {\em any } complex-valued object 
whose support has rank $\geq 2$ (Theorem \ref{thm:new}). We have proved the almost sure  uniqueness, up to a global phase,  under the two-point assumption (Theorem \ref{cor1}).   The absolute uniqueness is then enforced by the positivity constraint (Corollary \ref{cor}).  Under the tight  sector constraint, we have proved the absolute uniqueness 
with probability exponentially close to unity as the object sparsity increases  (Theorem \ref{thm4}).  Under the magnitude constraint, we have proved uniqueness up to
a global phase with probability exponentially close to unity (Theorem \ref{thm6}).  For general complex-valued objects without any constraint, we have established almost sure uniqueness modulo global phase with two independent  illuminations
(Theorem \ref{thm5}). 

Numerical experiments  reveal that
phasing with random illumination drastically reduces 
 the reconstruction error, the number of Fourier magnitude data and removes the stagnation 
 problem commonly associated with the ER and HIO
 algorithms.  
 %This indicates that the numerical phasing  problems are largely attributable  to the lack of absolute uniqueness. 
 Enforcement of absolute uniqueness therefore appears to have a profound effect on the performance of the standard phasing algorithms.

%Practical implementation of our approach demands precise maneuver of illumination which can be expected to realize with advances  of  technology.    
Systematic  and detailed study of
phasing   in the presence of (additive or multiplicative) 
noise with low-resolution random illuminations and
sub-Nyquist sampling rates  will be presented in
the forthcoming paper \cite{FL2}. 
\commentout{
Random amplitude illumination such as Gaussian or binary mask
has been used for phase retrieval in \cite{CES}.  Our numerical
experiments \cite{FL} have shown  a clear advantage of
random phase illumination over  the Gaussian mask
in reconstruction by the standard methods such as the Error-Reduction algorithm and  Hybrid Input-Output
algorithm \cite{Fie, Mar2}. 
}

%On the issue of convexity on the other hand,
%there have been recent attempts to formulate phase retrieval
%as a convex optimization problem \cite{CES, CMP}. 
%These approaches, however promising,  require a lot more Fourier magnitude data
%and computational resources.  

\commentout{
Finally, regarding noise stability, the set of reducible polynomials of more than one variable is contained in a nontrivial algebraic set
and its topological closure has zero measure \cite{San2}. 
 This result can be generalized to our setting with complex-valued objects and  random illumination,  and implies the stability of the 
irreducibility result  with respect to noise in the data. Numerical
experiments \cite{FL2} have indeed shown  an enhanced level
of noise stability in phasing  with random illumination. 
}

% phase retrieval is  ill-conditioned in the sense that  an admissible data vector perturbed by a small amount of random noise  is almost surely not admissible.

\commentout{
The previous remarks apply, in particular, to discrete Fourier transform phase retrieval from magnitude if the dimension of the problem is larger than or equal to 2. These new results concerning the conditioning of this problem indicate that the local ill-conditioning shown above should be overcome by means of some regularization technique.
}

\commentout{
This problem can be dealt as follows. Let $\bx$ be
the object vector $\{f(\bn)\}$ according to certain ordering of the lattice $ \cN$ and let $\bX=\bx\bx^*$. Let $\bA=[\Psi_{jk}]$ with
\[
\Psi_{jk}=\exp{\lt[- {2\pi\im jk\over (N_1+1)(N_2+1)}\rt]}
\]
 be the oversampled Fourier matrix. Then the Fourier magnitude
 data vector consists of  the diagonal elements of the matrix
 $\bA \bX\bA^*$, i.e. 
 \beq
 \label{17}
 \by=\hbox{diag}\Big(\bA \bX\bA^*\Big).
 \eeq
 We shall symbolically write (\ref{17}) as $\by= {\mathscr A}\bX$. 
 Suppose $\by$ is perturbed by a noise vector  $\bh$ of
 size $\ep$, i.e. $\|\bh\|_2\leq \ep$. Since the noisy data vector $\bz=\by+\bh$ is almost surely not feasible, we consider
 the enlarged feasible set
 \beq
 \label{18}
\Big\{\bX: \bX\,\, \hbox{is a rank one matrix such that }\,\,  \|{\mathscr A} \bX-\bz\|_2\leq \ep\Big\}
 \eeq
 which contains at least one element.
 }

\appendix
\section{Proof of Theorem  \ref{thm:new}}
Our argument is based on \cite{Kup, Oss} and
can be extended to the case of more than two independent variables. 
For simplicity of notation, we present the proof for the case of two independent variables. 
\begin{proof}
First we state an elementary result from algebraic geometry
(see, e.g., \cite{Ueno}, page 65). 
\begin{proposition}
If a homogeneous polynomial $P(z_0,z_1,z_2)$ of (total) degree $\delta\geq 2$ is irreducible, then
$ P^\flat(z_1,z_2)\equiv P(1,z_1,z_2)$ is also irreducible with degree $\delta$.  
\label{prop1}
\end{proposition}
For a polynomial $Q(z_1,z_2)$ of degree $\delta$,
the expression 
\beq
\label{hom}
Q^\sharp(z_0,z_1,z_2)=z_0^\delta Q({z_1\over z_0}, {z_2\over z_0})
\eeq
defines a homogeneous polynomial of degree $\delta$
with the property $Q(z_1,z_2)=Q^\sharp(1,z_1,z_2)$. 
The process from $Q$ to $Q^\sharp$  is
called homogenization while the reverse process is
called dehomogenization. 
Homogenization, in conjunction with Proposition \ref{prop1}, is a useful tool for studying the question of
irreducibility. 

Let $\Sigma\subset \IN^2$ be a given support set satisfying
the assumptions of Theorem \ref{thm:new}.  We now show
that almost all homogeneous polynomials of 3 variables 
with the support
\beq
\label{102}
\Sigma^\sharp=\{(\delta-n_1-n_2, n_1, n_2): \bn=(n_1,n_2)\in \Sigma\}
\eeq
 are irreducible.  
 
We represent the set of all homogeneous polynomials
of degree $\delta$ by the projective space $\IP^{\nu_\delta}$
of dimension $\nu_\delta=\big({2+\delta\atop \delta}\big)-1$.
Each homogeneous coordinate of $\IP^{\nu_\delta}$ represents a  monomial of degree $\delta$. 
The homogeneous polynomials supported on $\Sigma^\sharp$ 
are  represented by  the projective subspace
\[
\IX=\{P\in \IP^{\nu_\delta}: p(\bn)=0, \forall \bn\not\in \Sigma^\sharp\}
\]
where $\{p(\bn)\}$ are the coefficients of $P$. 
Clearly $\IX$ is isomorphic to the projective space
 $\IP^{s}$. 
 Let $\IY\subset \IX$ denote the set of reducible
 homogeneous polynomials supported on $\Sigma^\sharp$. 
We claim (cf. \cite{Sha}, page 47)
\begin{proposition}
$\IY$ is  a closed subset of $\IX$ in the Zariski topology. 
\label{prop2}
\end{proposition}
A subset of a projective space 
 is closed in the Zariski topology  if and only if it is an algebraic variety, i.e. the common zero set 
 \[
 \{ F_1=F_2=\cdots =F_m=0\}
 \]
of  a finite number of homogeneous polynomials $F_1, \cdots, F_m$ of the homogeneous coordinates of the projective space. The Zariski topology is  much cruder than the metric
 topology. Indeed, a Zariski closed set is either
 the whole space or a measure-zero, nowhere-dense closed set in the metric topology as stated in the following (see, e.g.  \cite{Lan}, page 115).
 
 \begin{proposition}
 Any Zariski closed proper subset of a (real or complex) projective variety 
 has measure zero with respect to the standard measure on the projective variety. 
 \label{prop:zero}
 \end{proposition}
 
{\em Proof of Proposition \ref{prop2}.} 
 Let the projective spaces $\IP^{\nu_j} $ and
$\IP^{\nu_{\delta-j}}$ represent  the homogeneous polynomials of degree $j$ and $\delta-j$, respectively,
where $\nu_j=\big({2+j\atop j}\big)-1$
and $\nu_{\delta-j}=\big({2+\delta-j\atop \delta-j}\big)-1$.
Let $\IY_j\subset \IY$ be the set of points corresponding to
polynomials supported on $\Sigma^\sharp$  that split
into factors of degree $j$ and $\delta-j$. Clearly
$\IY=\cup_{j=1}^{\delta-1} \IY_j$ and we need only prove
that each $\IY_j$ is Zariski closed. 

Now the multiplication of two polynomials of degree $j$
and $\delta-j$ determines a regular (i.e. polynomial)  mapping
\[
\Phi: \IP^{\nu_j}\times \IP^{\nu_{\delta-j}} \longrightarrow \IP^{\nu_{\delta}}
\]
in the following way. 
Let $G(\bz)$ and $H(\bz)$ be homogeneous polynomials of degrees $j$ and $\delta - j$, respectively.
Let $\{g(\bn)\}$ and $\{h(\bn)\}$ be
the coefficients of $G$ and $H$, respectively. 
Then the coefficients of the image point  $\Phi (G, H)$ are given by
\beq
\label{2}
\Big\{\sum_{|\bn|=\delta-j} g(\mbm-\bn)h(\bn):  |\mbm|=m_0+m_1+m_2=\delta \Big\}. 
\eeq
In other words $\Phi$ is bilinear in $\{g(\bn)\}$ and $\{h(\bn)\}$
and thus is regular. 
Clearly we have
\[
\IY_j=\Phi\Big(\IP^{\nu_j}\times \IP^{\nu_{\delta-j}}\Big)\cap \IX.
\]
Since the product of projective spaces is a projective variety 
and the image of a projective variety under a regular mapping is Zariski closed \cite{Sha}, $\IY_j$ is a Zariski closed subset of $\IX$. \begin{flushright} $\square$\end{flushright}

Let $\Sigma=\{\bn_1,\bn_2,\cdots,\bn_S\}$ and
let the ensemble of polynomials corresponding
to $\{\lambda(\bn) f(\bn)\}$ be identified with
\[
\prod_{\bn\in \Sigma}f(\bn)\cV(\bn)=(f(\bn_1)\cV(\bn_1))\times (f(\bn_2) \cV(\bn_2)) \times\cdots \times (f(\bn_S)\cV(\bn_S))\subseteq \IR^{2S}
\]
where 
\[
f\cV=\{(xf_1-yf_2, xf_2+yf_1)\in \IR^2: 
f_1=\Re(f), f_2=\Im(f), (x,y)\in \cV\}.
\] 
Note that $f\cV$ is a real algebraic variety in $\IR^2$ 
if $\cV$ is also. In a similar vein, we now treat complex
projective space $\IX$ (resp. variety $\IY$)  as real projective 
space (resp. variety) of double dimensions and,  
by homogenization, we embed 
$
\prod_{\bn\in \Sigma}f(\bn)\cV(\bn)$ 
in $\IX$ and denote the resulting projective variety as $\IV\subset\IX$. 

Clearly  $\IY\cap \IV$ is a real algebraic subvariety in $\IV$. To show
$\IY\cap \IV$ is a measure-zero subset of $\IV$ we only need to show that $\IY\cap \IV \subsetneq \IV$ in view of   Proposition \ref{prop2} and \ref{prop:zero}. 

Following the suggestion in \cite{Kup}, we now prove
\begin{proposition}
\label{prop3}
$\IY\cap \IV\subsetneq \IV$. 
\end{proposition}

{\em Proof of Proposition \ref{prop3}.} It suffices to find
{\em one}  irreducible polynomial in $\IV$.  

The argument is based on two observations.
First the polynomial  
\beq
\label{103}
F(x,y,z)=ax^r+by^r+cz^r
\eeq
is irreducible for any positive integer $r$ and any nonzero
coefficients $a,b,c$. 
This follows from the fact that the criticality equations
$F_x=F_y=F_z=0$ have no solution in $\IP^2$
and thus the algebraic variety $F=0$ is non-singular.

Secondly, for any $\Sigma$ satisfying the assumptions
of Theorem \ref{thm:new}, there exists a set $T\subset \Sigma^\sharp$ of three points  which can be transformed
into $\{(r, 0, 0), (0, r, 0), (0, 0, r)\}$, the support of
(\ref{103}), under a rational map. 

We separate  the analysis of the second observation into two cases.  

{\em Case 1:} $(0,0)\in \Sigma$.  Then there are at least two
other points, say $(m,n), (p,q)$, belonging to $\Sigma$. Without loss of generality, we
assume $p+q=\delta$. Because $\Sigma$ has rank 2,
$mq-np\neq 0$.

We look for the rational mapping
\beq
\label{105}
z_1=x^{k_{11}}y^{k_{21}}z^{k_{31}},\quad
z_2=x^{k_{12}}y^{k_{22}}z^{k_{32}},\quad
z_0=x^{k_{13}}y^{k_{23}}z^{k_{33}}
\eeq
with $k_{ij}\in \IZ$ that maps 
the polynomial
\[
P(z_0,z_1,z_2)=cz_0^\delta+az_0^{\delta-m-n}z_1^mz_2^n+
bz_1^pz_2^q
\] to $F(x,y,z)$. 
%\beq
%\label{107}
%Q(x,y,z)=x^ay^bz^c F(x,y,z),\quad a,b,c\in \IZ.
%\eeq
 This  amounts to
a linear transformation from
the set of independent vectors
\[
(m,n,\delta-m-n),\quad (p,q,0),\quad (0,0,\delta)
\]
to the set $\{
(r,0,0), (0,r,0), (0,0, r)\}$.
This transformation can be accomplished by
the following matrix
\beq
\label{106}
r\lt(\begin{matrix}
m&p&0\\
n&q&0\\
\delta-m-n&0&\delta
\end{matrix}
\rt)^{-1}
=
{r\over \delta(mq-np)}
\lt(\begin{matrix}
q\delta&-n\delta&-q(\delta-m-n)\\
-p\delta&m\delta&p(\delta-m-n)\\
0&0&mq-np
\end{matrix}\rt)
\eeq
where the divisor is nonzero.
To ensure integer entries in (\ref{106})
we set 
\[
 r=\delta(mq-np)
 \]
  and obtain
 the transformation matrix
\beqn
\lt(\begin{matrix}
k_{11}&k_{12}&k_{13}\\
k_{21}&k_{22}&k_{23}\\
k_{31}&k_{32}&k_{33}
\end{matrix}
\rt)=
\lt(\begin{matrix}
q\delta&-n\delta&-q(\delta-m-n)\\
-p\delta&m\delta&p(\delta-m-n)\\
0&0&mq-np
\end{matrix}\rt)\eeqn

{\em Case 2:}  $(m,0), (0,n)\in \Sigma$ for some positive integers $m,n$. Then there is at least another point $(p,q)\in \Sigma$ such
that $(m,0), (0,n), (p,q)$ are not collinear, which means
$mn-np-mq\neq 0$.

Suppose $p+q=\delta$. 
Consider the polynomial
\[
P(z_0,z_1,z_2)=az_0^{\delta-m}z_1^m+bz_0^{\delta-n}z_2^n+
cz_1^pz_2^q.
\]

By the same analysis above the form (\ref{103}) can be achieved by
the transformation matrix
\beqn
\label{108}
r\lt(\begin{matrix}
m&0&p\\
0&n&q\\
\delta-m&\delta-n&0
\end{matrix}
\rt)^{-1}
=
{r\over \delta(mn-mq-np)}
\lt(\begin{matrix}
-q(\delta-n)&q(\delta-m)&-n(\delta-m)\\
p(\delta-n)&-p(\delta-m)&-m(\delta-n)\\
-pn&-mq&mn
\end{matrix}\rt)\eeqn
which has integer entries if $r$ is a multiple of
$\delta(mn-mq-np)$. 
With the choice
\[
r= \delta(mn-mq-np)
\]
 the transformation matrix
becomes
\beqn
\lt(\begin{matrix}
k_{11}&k_{12}&k_{13}\\
k_{21}&k_{22}&k_{23}\\
k_{31}&k_{32}&k_{33}
\end{matrix}
\rt)=
\lt(\begin{matrix}
-q(\delta-n)&q(\delta-m)&-n(\delta-m)\\
p(\delta-n)&-p(\delta-m)&-m(\delta-n)\\
-pn&-mq&mn
\end{matrix}\rt).
\eeqn

Suppose $n=\delta$. Consider the polynomial
\[
P(z_0,z_1,z_2)=az_0^{n-m}z_1^m+bz_2^n+c
z_0^{n-p-q}z_1^pz_2^q.
\]
 The form (\ref{103}) can be achieved by
the transformation matrix
\beqn
r\lt(\begin{matrix}
m&0&p\\
0&n&q\\
n-m&0&n-p-q
\end{matrix}
\rt)^{-1}
=
{r\over n(mn-mq-np)}
\lt(\begin{matrix}
n(n-p-q)&q(n-m)&n(m-n)\\
0&mn-mq-np&0\\
-pn&-mq&mn
\end{matrix}\rt).\eeqn
With 
\[
r= n(mn-mq-np)
\]
 the transformation matrix
becomes
\beqn
\lt(\begin{matrix}
k_{11}&k_{12}&k_{13}\\
k_{21}&k_{22}&k_{23}\\
k_{31}&k_{32}&k_{33}
\end{matrix}
\rt)=
\lt(\begin{matrix}
n(n-p-q)&q(n-m)&n(m-n)\\
0&mn-mq-np&0\\
-pn&-mq&mn
\end{matrix}\rt).
\eeqn

Suppose $m=\delta$. Consider the polynomial
\[
P(z_0,z_1,z_2)=az_1^m+bz_0^{m-n}z_2^n+c
z_0^{m-p-q}z_1^pz_2^q. 
\]
 The form (\ref{103}) can be achieved by
the transformation matrix
\beqn
r\lt(\begin{matrix}
m&0&p\\
0&n&q\\
0&m-n&m-p-q
\end{matrix}
\rt)^{-1}
=
{r\over m(mn-mq-np)}
\lt(\begin{matrix}
mn-mq-np&0&0\\
p(m-n)&m(m-p-q)&m(n-m)\\
-pn&-mq&mn
\end{matrix}\rt).\eeqn
With 
\[
r= m(mn-mq-np)
\]
 the transformation matrix
becomes
\beqn
\lt(\begin{matrix}
k_{11}&k_{12}&k_{13}\\
k_{21}&k_{22}&k_{23}\\
k_{31}&k_{32}&k_{33}
\end{matrix}
\rt)=\lt(\begin{matrix}
mn-mq-np&0&0\\
p(m-n)&m(m-p-q)&m(m-n)\\
-pn&-mq&mn
\end{matrix}\rt).
\eeqn

To conclude the proof of Proposition \ref{prop3}, in any above case, if the polynomial $P(z_0,z_1,z_2)$ is reducible (i.e. has
a non-monomial  factor), then we can write
$P=P_1P_2$ and
\beq
\label{111}
F(x,y,z)=P_1(z_0(x,y,z), z_1(x,y,z), z_2(x,y,z))P_2(z_0(x,y,z), z_1(x,y,z), z_2(x,y,z))
\eeq
where $P_1, P_2$ are non-monomial factors. 
Let $l$ be
the lowest (possibly negative) power in $x,y,z$ of $P_i(z_0(x,y,z), z_1(x,y,z), z_2(x,y,z)), i=1,2$. If $l\geq 0$, then
 the factorization (\ref{111}) implies that
 $F(x,y,z)$  has a non-monomial factor. If $l<0$, then  the factorization (\ref{111}) implies that
 $(xyz)^{-l} F(x,y,z)$   has a non-monomial factor. Either case contradicts the fact that  $F$ is irreducible. So $P$ is irreducible. The proof of Proposition \ref{prop3} is complete. 
\begin{flushright}$\square$\end{flushright}

Continuing the proof of Theorem \ref{thm:new},
we have from Propositions \ref{prop:zero} and \ref{prop3}
that $\IY\cap\IV$ is a measure-zero subset of $\IV$.
By dehomogenization and Proposition \ref{prop1}
reducible polynomials of a fixed support $\Sigma$ under the
assumptions of Theorem \ref{thm:new} comprise
a measure-zero subset of  all polynomials of
the same support. The proof of Theorem \ref{thm:new} is now complete.

\end{proof}

\bigskip

{\bf Acknowledgements.} I am grateful to  my colleagues Greg Kuperberg and Brian Osserman  for
inspiring  discussions 
%and outlining the ideas of
on the proof of Theorem \ref{thm:new}, an improvement of the earlier version which assumes convexity of the support. 
I thank my student Wenjing Liao for performing 
simulations and producing the figures.

\end{document}